\documentclass[a4paper,12pt]{article}

\usepackage{amsfonts}
\usepackage{pifont}
\usepackage[dvipsnames,usenames]{color}
\usepackage{bbm}
\usepackage{mathrsfs}
\usepackage{amsmath}
\usepackage{amsthm}
\usepackage{graphicx, float,epsfig}
\usepackage{fullpage,graphicx,url}
\usepackage{comment}
\usepackage[subfigure]{graphfig}
\usepackage{color}

\newcommand\B{{\bf B}}
\newcommand\bO{{\bf O}}
\newcommand\ba{{\bf a}}
\newcommand\bu{{\bf u}}
\newcommand\x{{\bf x}}

\newcommand\bc{{\bf c}}
\newcommand\e{{\bf e}}
\newcommand\bb{{\bf b}}
\newcommand\y{{\bf y}}

\newcommand\z{{\bf z}}
\newcommand\bv{{\bf v}}

\newcommand\bz{{\bf 0}}
\newtheorem{corollary}{Corollary}
\newtheorem{definition}{Definition}
\newtheorem{example}{Example}
\newtheorem{lemma}{Lemma}

\newtheorem{theorem}{Theorem}

\newtheorem{algorithm}{Algorithm}[section]
\definecolor{Red}{rgb}{1,0,0}
\definecolor{Blue}{rgb}{0,0,1}

\def\N={{\cal N} \rm}

\begin{document}

\title{\bf Computations and Complexities of Tarski's Fixed Points and Supermodular Games%\footnote{Some results in this paper are based on our working paper (2011) Computational Models and Complexities of Tarski's Fixed Points (\url{https://web.stanford.edu/~yyye/unitarski1.pdf}).}
}

\author{Chuangyin Dang
\\Dept. of Systems Engineering \& Engineering Management\\
City University of Hong Kong\\ Kowloon, Hong Kong SAR, China\\
E-Mail: mecdang@cityu.edu.hk\\
\\
Qi Qi\\
Dept. of Industrial Engineering \& Decision Analytics \\
The Hong Kong University of Science and Technology\\
 Kowloon, Hong Kong SAR, China
\\E-Mail: kaylaqi@ust.hk\\
\\
Yinyu Ye\\
Dept. of Management Science \& Engineering \\
Stanford University\\
 Stanford, CA  94305-4026
\\E-Mail: yinyu-ye@stanford.edu}

\date{ }

\maketitle

\begin{abstract}
%\setlength{\baselineskip}{8mm}
%It has been shown in the literature that, when $f$ is given
%as an oracle, a fixed point of $f$ can be found in {\em polynomial time}.
We consider two models of computation for Tarski's order preserving function $f$
related to fixed points in
a complete lattice: the oracle function model and the polynomial function model. In both models, we find the first polynomial time algorithm for finding a Tarski's fixed point. In addition, we provide a matching oracle bound for determining the uniqueness in the oracle function model and prove it is Co-NP hard in the polynomial function model. The existence of the pure Nash equilibrium in supermodular games is proved by Tarski's fixed point theorem %\cite{tarski1}. 
Exploring the difference between
supermodular games and Tarski's fixed point, we also develop the computational results for finding one pure Nash equilibrium and determining the uniqueness of the equilibrium in supermodular games.

%
%We develop a complete understand under the oracle function model for finding a Tarski's
%fixed point as well as determining uniqueness of the Tarski's fixed point
%in both the lexicographic ordering and the componentwise ordering lattices.
%In addition, we also present a
%polynomial-time reduction of an integer programming to an order preserving
%mapping $f$ from a lattice $L$ into itself. As a result of this reduction, we prove
%that, when $f$ is given as a polynomial function, determining whether or not $f$ has a
%unique fixed point is {\em Co-NP hard}.
\end{abstract}


\begin{section}{Introduction}

Supermodular games, also known as the games of strategic complements,  are formalized by Topkis in 1979 \cite{t1} and have been extensively studied in the literature,  such as Bernstein and Federgruen \cite{b2}\cite{b1}, Cachon \cite{c1}, Cachon and Lariviere \cite{c2}, Fudenberg and Tirole \cite{fudenberg}, Lippman and McCardle \cite{l1}, Milgrom and Roberts
\cite{milgrom0}\cite{milgrom1}, Milgrom and Shannon \cite{milgrom2}, Topkis \cite{topkis1}, and Vives
\cite{vives0}\cite{vives2}. In supermodular games, the utility function of every player has increasing differences.  Then the best response of a player is a nondecreasing function of other players' strategies. For example, if firm A's competing firm B starts spending more money on research it becomes more advisable for firm A to do the same.

Supermodular games arise in many applied models. They cover most static market models. For example, the investment games, Bertrand oligopoly, Cournot oligopoly all can be modeled as supermodular games. Many models in operations research have also been analyzed as supermodular games. For example, supply chain analysis, revenue management games, price and service competition, inventory competition etc. Recently, the problem of power control in cellular CDMA wireless network is also modeled as a supermodular game.

The existence of a pure Nash equilibrium in any supermodular game is proved by Tarski's fixed point theorem \cite{tarski1}. The well-known Tarski's fixed point theorem (Tarski) asserts that, if $(L, \preceq)$ is a complete lattice and $f$ is order-preserving from $L$ into itself, then there exists some $\x^*\in L$ such that
$f(\x^*)=\x^*$.

%A partially order set $L$ is defined
% with $\preceq$ as a binary relation on the set $L$ such that
% $\preceq$ is reflexive,
%transitive, and anti-symmetric.
%A lattice is a partially ordered set
%$(L,\preceq)$, in which any two elements $\x$ and $\y$ have a least
%upper bound (supremum), $\sup_L(\x,\y)=\inf\{\z\in L\;|\;\x\preceq
%\z\mbox{ and }\y\preceq \z\}$, and a greatest lower bound (infimum),
%$\inf_L(\x,\y)=\sup\{\z\in L\;|\;\z\preceq \x\mbox{ and }\z\preceq \y\}$,
%in the set. A lattice $(L,\preceq)$ is complete if every nonempty
%subset of $L$ has a supremum and an infimum in $L$.
%Let $f$ be a
%mapping from $L$ to itself. $f$ is order-preserving if
%$f(\x)\preceq
% f(\y)$ for any $\x$ and $\y$ of $L$ with $\x\preceq \y$.
%

This theorem plays a crucial role in the study
of supermodular games for economic analysis and has other important applications.
To compute a Nash equilibrium of a supermodular game,
a generic approach is to convert it into the computation of a fixed
point of an order preserving mapping. Recently, an algorithm has been
proposed in Echenique \cite{e1} to find all pure strategy Nash
equilibria of a supermodular game, which motivated to the study in this
paper.

An efficient computational algorithm for finding a Nash equilibrium has been a recognized important technical advantage in applications.
Further, it is sometimes desirable to know if an already-found equilibrium for such applications is unique or not, for the decision whether additional resource should be spent to improve the already found solution.
There were some interesting complexity results in algorithmic game theory
research along this line, on determining whether or not a game has a unique equilibrium point.
For the bimatrix game, Gilboa and Zemel \cite{gilboa} showed that
it is NP-hard to determine whether or not there is a second Nash equilibrium.
For this problem, computing even one
equilibrium (which is know to exist), is already difficult and no polynomial time
algorithms are known: Nash equilibrium for the bimatrix game is known to be PPAD-complete~\cite{CDT}.
Similar cases are known for other problems such as the
market equilibrium computation (Codenotti et al.)\cite{CSVY}.

In this work, we first consider the fixed point computation of order preserving functions over a complete lattice, both for finding a solution and for determining the uniqueness of an already-found solution. Then we study the computational problems for finding one pure Nash equilibrium and determining the uniqueness of the equilibrium in supermodular games.
We are interested in both the oracle function model and the polynomial function model.  For both the fixed point problem and supermodular games, the domain space can be huge. Most interesting discussions consider a succinct representation (see Section \ref{succinct}) of the lattice $(L,\preceq)$ such that the input size is related to $\log |L|$. It is enough for the representation of a variable in a lattice of size $|L|$. Both the oracle function model and the polynomial time function model return the function value $f(x)$ on a lattice node $x$ where $x$ is of size $\log |L|$. They differ in the ways the functions are computed. The polynomial time function model computes $f(x)$ by an explicitly given algorithm, in time polynomial of $\log |L|$. The oracle model, on the other hand, always returns the value in one oracle step. More details comparing those two models can be found in Section \ref{two function}.

\begin{subsection}{Main Results and Related Work}

A partially order set $L$ is defined
 with $\preceq$ as a binary relation on the set $L$ such that
 $\preceq$ is reflexive,
transitive, and anti-symmetric.
A lattice is a partially ordered set
$(L,\preceq)$, in which any two elements $\x$ and $\y$ have a least
upper bound (supremum), $\sup_L(\x,\y)=\inf\{\z\in L\;|\;\x\preceq
\z\mbox{ and }\y\preceq \z\}$, and a greatest lower bound (infimum),
$\inf_L(\x,\y)=\sup\{\z\in L\;|\;\z\preceq \x\mbox{ and }\z\preceq \y\}$,
in the set. A lattice $(L,\preceq)$ is complete if every nonempty
subset of $L$ has a supremum and an infimum in $L$.
Let $f$ be a
mapping from $L$ to itself. $f$ is order-preserving if
$f(\x)\preceq
 f(\y)$ for any $\x$ and $\y$ of $L$ with $\x\preceq \y$.

We focus on the componentwise ordering and lexicographic ordering finite lattices.
Let $L_d=\{\x\in Z^d\;|\;\ba\leq \x\leq \bb\}$, where $\ba$ and $\bb$
are two finite vectors of $Z^d$ with $\ba<\bb$. We denote the componentwise ordering and the lexicographic ordering as $\leq_c$ and $\leq_l$ respectively. Clearly,  $(L_d, \leq_c)$ is a finite lattice with componentwise ordering and $(L_d,\leq_l)$ is a finite lattice with lexicographic ordering.

Let $f_c$ and $f_l$ be an order preserving mapping from $L_d$ into
itself under the componentwise ordering and the lexicographic ordering respectively.

\begin{subsubsection}{Tarski's Fixed Points: Oracle Function Model}

When $f_l(\cdot)$ and $f_c(\cdot)$ are given as oracle functions, we develop a complete understand for finding a Tarski's
fixed point as well as determining uniqueness of the Tarski's fixed point
in both the lexicographic ordering and the componentwise ordering lattices.

We develop an algorithm of time complexity $O((\log^d |L|))$
to find a Tarski's fixed point on the componentwise ordering lattice $(L,\le_c)$, for any constant dimension $d$. This algorithm is based on the binary search method. We first present the algorithm when $d=2$. Follows the similar principle, this algorithm can be generalized to any constant dimension. This is the first known polynomial time algorithm for finding the Tarski's fixed point in terms of the componentwise ordering. In literature, we only have a polynomial time algorithm for the total order lattices (Chang et al.)~\cite{chang1}.

Recently, Mihalis, Kusha and Papadimitriou stated in a private communication that they proved a lower bound of $\Omega(\log^2|L|)$ in the oracle function model for finding a TarskiÕs fixed point in the two dimensional case, and conjectured a lower bound of $\Omega(\log^d|L|)$ for general $d$  (Christos H. Papadimitriou, private communication, March, 2019). Together with our upper bound results, they establish a matching bound of ${\Theta(\log^2|L|)}$ for finding a TarskiÕs fixed point in the two dimensional case. 

On the other hand, given a general lattice $(L,\preceq)$ with one already known fixed point, to find out whether it is unique will take $\Omega(|L|)$ time for any algorithm. For componentwise ordering lattice, we derive a $\Theta(N_1+N_2+\cdots+N_d)$ matching bound for determining the uniqueness of the fixed point, where $L=\{\x\in Z^d\;|\;\ba\leq \x\leq \bb\}$ and $N_i=b_i-a_i$. In addition, we prove this matching bound for both deterministic algorithm and randomized algorithm.

For a lexicographic ordering lattice, it can be viewed as a componentwise ordering lattice with dimension one by an appropriate polynomial time transformation to  change the oracle function for the $d$-dimension space to an oracle function on the $1$-dimension space.
All the above results can be transplanted onto the lexicographic ordering lattice with a set of related parameters.

In literature, a polynomial time algorithm is known only for the total order lattices.
When the lattice $(L,\preceq)$ has a total order, i.e., all the point in the lattice is comparable, there is a matching bound of $\theta(\log |L|)$, where an $\Omega(|L|)$ lower bound is known for general lattices (when the lattice is given as an oracle) in Chang et al.\cite{chang1}.

\end{subsubsection}

\begin{subsubsection}{Tarski's Fixed Points: Polynomial Function Model}

Under the polynomial time function model, our polynomial time algorithm applies when the dimension is any finite constant.
When the dimension is used as a part of the input size in unary, we first
present a polynomial-time reduction of a 3-SAT problem to an order preserving mapping $f$ from a componentwise ordering lattice $L$ into itself. As a
result of this reduction, we obtain that, given $f$ as a polynomial time
function, determining whether $f$ has a unique fixed point in $L$ is a Co-NP
hard problem. Furthermore, even when the dimension is one, we also show that determining the uniqueness of Tarski's fixed point in a lexicographic lattice is Co-NP hard though there exists a polynomial-time algorithm for
computing a Tarski's fixed point in a lexicographic lattice in any dimension.

Our main results for Tarki's fixed point computation are summarized in Table \ref{one Tarski result} and Table \ref{unitarski result}.
\end{subsubsection}

\begin{subsubsection}{Supermodular Games}
For supermodular games, we develop an algorithm to find a pure Nash equilbirum in polynomial time $O(\log N_1\cdots\log N_{d-1})$ in the oracle function model, where $d$ is the total number of players, $N_i$ is the number of strategies of player $i$ and $N_1\le N_2\cdots\le N_d$. It is the first polynomial time algorithm when $d$ is a constant. Thus a pure Nash equilibirum can be found in time $O(poly(\log|L|)\cdot(\log N_1\cdots\log N_{d-1})$ in the polynomial function model, where $|L|=N_1\times N_2\cdots\times N_d$. % Given a pure Nash equilibrium, we find an oracle matching bound $\Theta(N_1+\cdots+N_{d-1})$  for determining whether it is unique or not. 
In the polynomial function model, we prove determining the uniqueness is Co-NP-hard. %The main results for supermodular game computation are listed in Table \ref{supermodular result}.

In literature, Robinson(1951) \cite{robinson} introduce the iterative method to solve a game and Topkis(1979) \cite{t1} use this method to find a pure Nash equilibrium in supermodular game which takes time $O(N_1+N_2+\cdots+N_d)$. The first non-trivial algorithm for finding pure Nash equilibria is proposed by Echenique in 2007\cite{e1}. However, the algorithm takes expenontial time $O(N1\times N2\times\cdots\times N_d)$ to find the first pure equilibrium in the worst case.

\end{subsubsection}

\begin{table}[H]
\centering
\begin{tabular}{|l|l|l|}
\hline
& Polynomial Function&Oracle Function\\
\hline
Componentwise&$O(poly(\log|L|)\cdot(\log N_1\cdots \log N_d))$ &$O((\log N_1\cdots \log N_d)$ \\
\hline
Lexicographic &$O(poly(\log|L|)\cdot\log|L|)$ & $O(\log|L|)$\\

\hline
\end{tabular}
\caption{Main Results for Finding one Tarski's Fixed Point}\label{one Tarski result}
\end{table}

\begin{table}[H]
\centering
\begin{tabular}{|l|l|l|}
\hline
& Polynomial Function&Oracle Function\\
\hline
Componentwise&Co-NP-Complete &$\Theta(N_1+N_2+\cdots+N_d)$ \\
\hline
Lexicographic &Co-NP-Complete & $\Theta(|L|)$\\

\hline
\end{tabular}
\caption{Main Results for Determining the Uniqueness of Tarski's Fixed Points}\label{unitarski result}
\end{table}

%\begin{table}[H]
%\centering
%\begin{tabular}{|l|l|l|}
%\hline
%&Polynomial Model & Oracle Model\\
%\hline
%Finding one equilibrium & $O(poly(\log|L|)\cdot(\log N_1\cdots\log N_{d-1}))$ & $O(\log N_1\cdots\log N_{d-1})$ \\
%\hline
%Uniqueness  & Co-NP-Complete &$\Theta(N_1+\cdots+N_{d-1})$\\
%\hline
%\end{tabular}
%\caption{Main Results for Supermodular Games}\label{supermodular result}
%\end{table}

\end{subsection}

%\begin{subsection}{Related Work}
%
%For the oracle model,
%when the lattice $(L,\preceq)$ has a total order, i.e., all the point in the lattice is comparable, there is
%an upper bound of $O(\log |L|)$, where an $\Omega(|L|)$ lower bound is known for general
%lattices (when the lattice is given as an oracle) to find a Tarski's fixed point on the lattice $(L,\preceq)$ in Chang et al.\cite{chang1}.
%\end{subsection}

\begin{subsection}{Organization}
The rest of the paper is organized as follows.
First, in Section 2, we present definitions as well as the difference of the polynomial function model and the oracle function model.
We develop polynomial time algorithms in oracle function model for componentwise ordering and lexicographic ordering in Section 3.
In Section 4, we derive the matching bound for determining the uniqueness of Tarski's fixed point under the oracle function model.
We prove co-NP hardness for determining the uniqueness of Tarski's fixed point under the polynomial function model in Section 5. In Section 6, we develop the computational results for finding one pure Nash equilibrium and determining the uniqueness of the equilibrium in supermdular games.
We conclude with discussion and remarks on our results and open problems in Section 7.
\end{subsection}
\end{section}

\begin{section}{Preliminaries}
In this section, we first introduce the formal definitions of the related concepts as well as the Tarski's fixed point theorem. We next compare the difference between the oracle function model and the polynomial function model.
\begin{subsection}{The Lattice and Tarski's Fixed Point Theorem}

\begin{definition}(Partial Order vs. Total Order)
A relationship $\preceq$ on a set $L$ is a partial order if it satisfies reflexivity
($\forall \ba\in L: \ba\preceq \ba$); antisymmetry ($\ba\preceq \bb$ and $\bb\preceq \ba$ implies
$\ba=\bb$); transitivity ($\ba\preceq \bb$ and $\bb\preceq \bc$ implies $\ba\preceq \bc$).
It is a total order if $\forall \ba, \bb\in L$: either $\ba\preceq \bb$ or $\bb\preceq \ba$.
\end{definition}

\begin{definition}(Lattice)
$(L, \preceq)$ is a lattice if
\begin{enumerate}
\item $L$ is a partial ordered set;
\item There are two operations: meet $\wedge$ and join $\vee$ on any pair of elements $\ba,\bb$ of $L$
such that $\ba,\bb\preceq \ba\vee \bb$ and $\ba\wedge \bb\preceq \ba,\bb$
\end{enumerate}
The lattice is complete lattice if for any subset $A=\{\ba_1,\ba_2,\cdots, \ba_k\}\subseteq L$, there is a unique meet and
a unique join: $\bigwedge A = (\ba_1\wedge \ba_2\wedge\cdots\wedge \ba_k)$
and
$\bigvee A = (\ba_1\vee \ba_2\vee\cdots\vee \ba_k)$.
\end{definition}
For simplicity, we use $L$ for a lattice when no ambiguity exists on $\preceq$.
We should specify $\preceq$ whenever it is necessary.

\begin{definition}(Order Preserving Function)
A function $f$ on a lattice $(L,\preceq)$ is order preserving if
$\ba\preceq \bb$ implies $f(\ba)\preceq f(\bb)$.
\end{definition}

\begin{theorem}(Tarski's Fixed Point Theorem)\cite{tarski1}.
If $L$ is a complete lattice and $f$ an increasing from $L$ to itself, there exists some $\x^*\in L$ such that $f(\x^*)=\x^*$, which is a fixed point of $f$.
\end{theorem}
This theorem guarantees the existence of fixed points of any order-preserving
function $f:L\rightarrow L$ on any nonempty complete lattice.

\begin{definition}(Lexicographic Ordering Function).
Given a set of points on a $d$-dimensional space $R^d$, the lexicographic ordering function $\leq_l$ is defined as:

$\forall \x,\y\in R^d$, $\x\leq_l \y$ if
either $\x=\y$ or $x_i=y_i$, for $i=1,2,\ldots,k-1$, and $x_k<y_k$ for
some $k\le d$.
\end{definition}

\begin{definition}(Componentwise Ordering Function).
Given a set of points on a $d$-dimensional space, the componentwise ordering function $\leq_c$ is defined as:

$\forall \x,\y\in R^d$, $\x\leq_c \y$ if $\forall i\in\{1,2,\cdots, d\}: x_i\leq y_i$.
\end{definition}
\end{subsection}

\subsection{Big Input Data and Succinct Representation}\label{succinct}

For the problems we consider in this work, there are usually $2^{d*n}$ nodes where $d$ is a constant and $n$
is an input parameter. Therefore, the input size is exponential in the
input parameter $n$. We need to represent such input data succinctly.
As an example,
for the set $N=\{0,1,2,\cdots, 2^n-1 \}$, the input can be described
as all the integers $i$: $0\leq i\leq 2^n-1$.
Each such integer $i$ can be written by up to $n$ bits.

When a computational problem involved a function such as $f:N\rightarrow N$.
There is always a question how this function is given as an input?
As an exmaple, let $f(i)$ represent the parity of the integer $i\in N$.
Then as an input to the computational problem, $f$ can be an circuit that
takes the last bit of the input $i$.
Therefore, the size of $f$ is a polynomial in the number of bits, $n$, of the
input data.
In general, however, the input functions are not that simple.
We should define two models of functions for succinct representation
with input involved with functions on big dataset in our computational
problems.

\subsection{The Oracle Function Model Versus the Polynomial Time Function Model}\label{two function}

The two succinctly represented function models are the
oracle functions and the polynomial functions.

For the oracle model, we treat the function as a black box that
outputs the function value for every domain variable once a request
is sent in to the oracle. The output of the oracle is arbitrary
on the first query but it cannot change a function value after a query
is already made to the oracle on the same variable.
For exmaple, let $N=\{0,1\}$. Let $f:N\rightarrow N$ be an
oracle function.
When we ask for $f(0)$, the oracle could answer anything,
either $0$ or $1$.
Suppose the oracle answers $f(0)=1$ in the first query in
one run of our algorithm. Later, if we need to use $f(0)$
again in the same run of the algorithm, it must be the same
$1$. Equivalently, we may assume that the function values
are stored in the harddisk. After a query, it is saved in the
memory cache. Later uses of the same query will be the value
in the memory cache and there is no need to check with the harddisk
again.
It is important to note that, the oracle funciton model contains
all the functions $f: N\rightarrow M$ where $N$ is its domain
and $M$ is its range. This is very different from the polynomial
function we are going to introduce next.

For the polynomial function model, the input function
is an algorithm that gives the answer for the function value
on the input data.
The algorithm returns the answer in time polynomial in the input parameter $n$.
Alternatively, the polynomial time algorithm can be replaced
by a polynomial size logical circuits consisting of
gates $\{AND, NOT, OR\}$ of Boolean variables.

Clearly oracle function admits much more functions than those computable
in polynomial time. Therefore a problem is usually much harder
under the oracle function model than under the polynomial time
function model.

\end{section}

\begin{section}{Polynomial Time Algorithm under Oracle Function Model}

In this section, we consider the complexity of finding a Tarski's fixed point in any constant dimension $d$ with the function value $f$ given by an oracle. Chang et al. \cite{chang1} proved that a fixed point can be found in time polynomial when the given lattice is total order.

Define $L=\{\x\in Z^d\;|\;\ba\leq \x\leq \bb\}$, where $\ba$ and $\bb$
are two finite vectors of $Z^d$ with $\ba<\bb$.

\begin{theorem}(Chang et al.)\cite{chang1}
When $(L,\preceq)$ is given as an input and the order preserving function $f$ is given as an oracle, a Tarski's fixed point can be found in time $O(\log |L|)$ on a finite lattice when $\preceq$ is a total order on $L$.
\end{theorem}

Since any two vectors in the lexicographic ordering is comparable, the lexicographic ordering is a total order. We have
\begin{corollary}
When $(L,\preceq)$ is given as an input and the order preserving function $f$ is given as an oracle, a Tarski's fixed point can be found in time $O(\log |L|)$ on a finite lattice when $\preceq$ is a lexicographic ordering in $L$.
\end{corollary}

  The proof is rather standard utilizing the total order property of the lexicographic ordering. As the componentwise ordering lattice cannot be modelled as a total order, it leaves open the oracle complexity of finding a fixed point in componentwise ordering lattice. Here we show that this problem is also polynomial time solvable, by designing a polynomial algorithm to find a fixed point of $f$ in time $O((\log |L|)^d)$ given componentwise ordering lattice $L$.

The algorithm exploits the order properties of the componentwise lattice and applying the binary search method with a dimension reduction technique. To illustrate the main ideas, we first consider the 2D case before moving on to the general case.

WLOG, we assume $L$ is a $N\times N$ square centred at point $(0,0)$. The componentwise ordering is denoted as $\le_c$.

\begin{algorithm}Point\_check() (A polynomial algorithm for 2D lattice)\label{2D}

\begin{itemize}
\item Input:
\begin{description}
\item
2-dimensional lattice $(L,\leq_c)$, $|L|=N^2$ (Input size to the oracle is $2\log N$ since the input size for both dimensions to the oracle is $\log N$. )
\item
Oracle function $f$. $f$ is a order preserving function. $\forall \x\in L,f(\x)\in L$ and $f(\x)\leq_c f(\y)$ if $\x\leq_c \y, \forall \x,\y\in L$

\end{description}
\item Point\_check($L,f$)
\begin{description}
\item Let $\x^0$ be the center point in $L$. Let $\x^L$ be the left most point in $L$ such that $\x^L_2=\x^0_2$. Let $\x^R$ be the right most point in $L$ such that $\x^R_2=\x^0_2$.
\begin{enumerate}
\item If $f(\x^0)=\x^0$,return($\x^0$);end;
\item If $f(\x^0)\geq_c \x^0$,$L'=\{\x|\x\geq_c \x^0,\x\in L\}$. Point\_check($L',f$);
\item If $f(\x^0)\leq_c\x^0$, $L''=\{\x|\x\leq_c \x^0,\x\in L\}$. Point\_check($L'',f$);
\item If $f(\x^0)_1<x^0_1$ and $f(\x^0)_2>x^0_2$, Binary\_Search($\x^L,\x^0$);
\item If $f(\x^0)_1>x^0_1$ and $f(\x^0)_2<x^0_2$, Binary\_Search($\x^0,\x^R$);
\end{enumerate}
\end{description}

\item Binary\_Search($\x,\y$)
\begin{description}
\item Let $\x^m=\lfloor{1/2(\x+\y)}\rfloor$
\begin{enumerate}
\item If $f(\x^m)=\x^m$,return($\x^m$);end;
\item If $f(\x^m)\geq_c \x^m$, $L'=\{\x|\x\geq_c \x^m,\x\in L\}$. Point\_check($L',f$);
\item If $f(\x^m)\leq_c\x^m$, $L''=\{\x|\x\leq_c \x^m,\x\in L\}$. Point\_check($L'',f$);
\item If $f(\x^m)_1<x^m_1$ and $f(\x^m)_2>x^m_2$, Binary\_Search($\x,\x^m$);
\item If $f(\x^m)_1>x^m_1$ and $f(\x^m)_2<x^m_2$, Binary\_Search($\x^m,\y$);
\end{enumerate}
\end{description}

\end{itemize}

\end{algorithm}

\begin{theorem}
When the order preserving function $f$ is given as an oracle, a Tarski's fixed point can be found in time $O(\log^2N)$ on a finite 2D lattice formed by integer points of a box with side length $N$ by using Algorithm \ref{2D} Point\_check.
\end{theorem}

\begin{proof}
Start from a lattice of size $|L|$, we first prove that in at most $O(\log N)$ steps the above algorithm either finds the fixed point or reduces the input lattice to size $|L|/2$.

\begin{Figure}[H]{A polynomial algorithm for 2D Lattice}
    \graphfile[50]{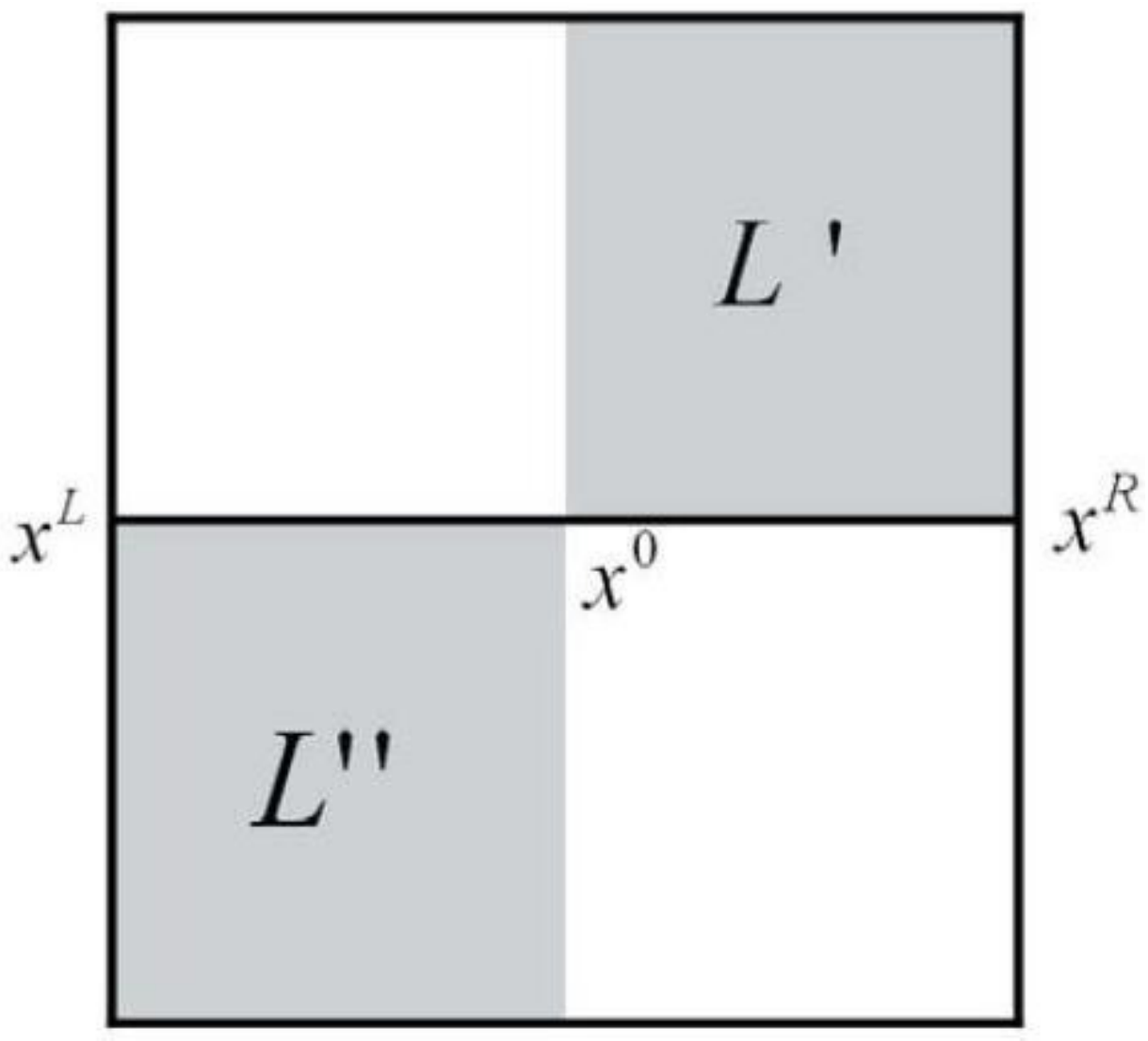}[Point\_check()]
    \graphfile[50]{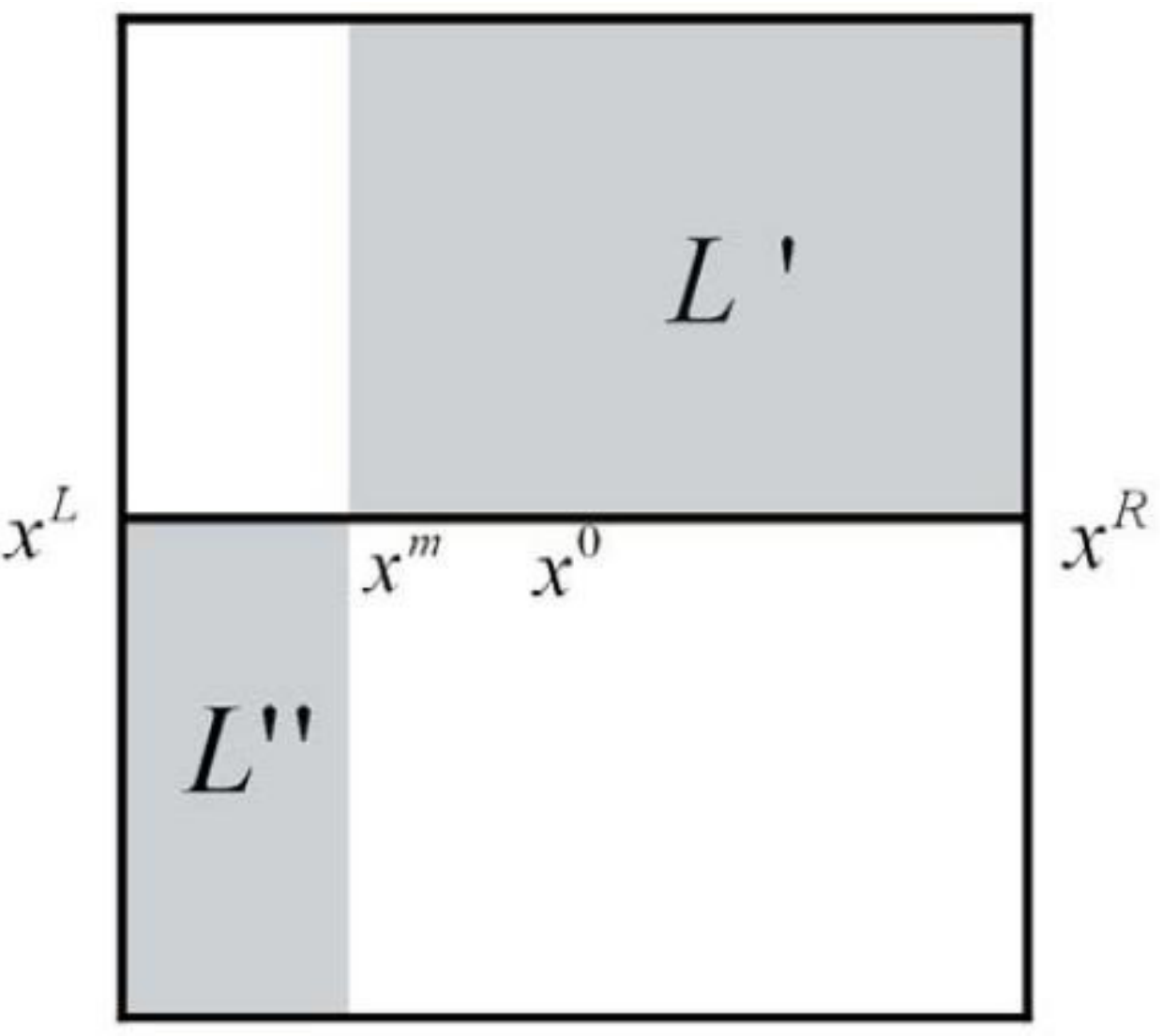}[Binary\_Search()]\label{alg1}
\end{Figure}
Consider the algorithm Point\_check($L,f$).
\begin{enumerate}
\item Case I: If $f(\x^0)=\x^0$, $\x^0$ is the fixed point which is found in 1 step.
\item Case II: If $f(\x^0)\geq_c \x^0$, since $f$ is a order preserving function, $\forall \y\geq_c \x^0$, we have $f(\y)\geq_c f(\x^0)\geq_c \x^0$. Let $L'=\{\x|\x\geq_c \x^0,\x\in L\}$. Define $f'(\x)=f(\x),\forall \x\in L'$. Then $f': L'\rightarrow L'$ is a order preserving function on the complete lattice $L'$. By Tarski's fixed point theorem, there must exist a fixed point in $L'$. Next we only need to check $L'$ which is only $1/4$ size of $|L|$.
\item Case III: If $f(\x^0)\leq_c\x^0$, similar to the analysis in Case II, we only need to consider $L''=\{\x|\x\leq_c \x^0,\x\in L\}$ which is only $1/4$ size of $|L|$ in the next step.
\item Case IV: If $f(\x^0)_1<x^0_1$ and $f(\x^0)_2>x^0_2$, we prove that Binary\_Search($\x^L,\x^0$) finds a fixed point or reduce the size of the lattice by half in $\log {\frac{N}{2}}$ steps. Since $f$ is a order preserving function, $\forall$ adjacent points $\bu\leq_c\bv\in L$, it is impossible that $f(\bu)_1>u_1$ and $f(\bv)_1<v_1$. Thus, on a line segment $[\x,\y]$ where $x_2=y_2$, if $f(\x)_1\ge x_1$ and $f(\y)_1< y_1$,  there must exist a point $\z$ such that $f(\z)_1= z_1$.  On the other hand, we have $f(\x^0)_1<x^0_1$ and by the boundary condition $f(\x^L)_1\ge x^L_1$, therefore, there must exist a point $\x'\in[\x^l,x^0)$ such that $f(\x')_1=x'_1$. This point $\x'$ can be found in time $\log {\frac{N}{2}}$ by using binary search. If $f(\x')_2>x'_2$, similar to the analysis in Case II, we only need to consider $L'=\{\x|\x\geq_c \x',\x\in L\}$ which is at most $1/2$ size of $|L|$ in the next step. If $f(\x')_2<x'_2$, we only need to consider $L''=\{\x|\x\leq_c \x',\x\in L\}$ which is at most $1/4$ size of $|L|$ in the next step. If $f(\x')_2=x'_2$, then $\x'$ is the fixed point.

\item Case V: If $f(\x^0)_1>x^0_1$ and $f(\x^0)_2<x^0_2$, similarly, we can prove that Binary\_Search($\x^0,\x^R$)finds a fixed point or reduce the size of the lattice by half in $\log {\frac{N}{2}}$ steps.
\end{enumerate}

The size of the lattice is reduced by half in every $O(\log N)$ steps. Therefore, the algorithm finds a fixed point in at most $O(\log N\times\log L)=O(\log^2 N)$ steps.\qed

\end{proof}

The above algorithm can be generalized to any constant dimensional lattice with $L=\{\x\in Z^d\;|\;\ba\leq \x\leq \bb\}$, where $\ba$ and $\bb$
are two finite vectors of $Z^d$ with $\ba<\bb$.
We reduce a $(d+1)$-dimension problem to a $d$-dimension one.
Assume we have an algorithm for a $d$-dimensional problem with time complexity $O(\log^d |L|)$. Let the algorithm be $A_d(L, f)$.

Consider a $d+1$-dimensional lattice $(L,\leq_c)$. Choose the central point in $L$, and denote it by $\bO=(O_1,O_2,\cdots,O_{d+1})^T$.
Take the section of $L$ by a hyperplane parallel to $x_{d+1}=0$ passing through $\bO$. Denote it as $L^d$.
Clearly, it is a $d$-dimensional lattice. We define a new oracle function $f_d$ on $L^d$, based on the oracle function $f$ on $L$.
Define $f_d(x_1,x_2,\cdots, x_d)=(y_1,y_2,\cdots, y_d)$, if $f(x_1,x_2,\cdots, x_d, O_{d+1})=(y_1,y_2,\cdots, y_d, y_{d+1})$.
We apply the algorithm $A_d(L^d, f_d)$ to obtain a Tarski's fixed point in time $(\log |L|)^d$. Let the fixed point be denoted by $\x^*$.
Therefore, $f(\x^*) = (\x^*, O_{d+1}) + a\e_{d+1}$ or $f(\x^*)=  (\x^*, O_{d+1}) - a\e_{d+1}$, where $a$ is some constant, $\e_{d+1}$ is a $d+1$ dimensional unit vector with 1 on its $d+1$th position.

In either case, we obtain a new box $\B$ with size no more than half of the original box defined by $[\ba,\bb]$, such that $f(\cdot)$ maps
all points in $\B$ into $\B$ and is order preserving. We can apply the algorithm recursively on $\B$.
The base case can be handle easily.
Therefore the total time is

$$T( |L|^{d+1}) \leq T({|L|^{d+1}\over 2}) + O (\log^d |L|).$$

It follows that $T(|L|^{d+1}) = O(\log^d |L|)$.

Formally, the polynomial time algorithm for finding a Tarski's fixed point in a d-dimensional componentwise ordering lattice is described as follows.
\begin{algorithm}Fixed\_point() (A polynomial algorithm for any constant dimensional lattice)\label{dD}

\begin{itemize}
\item Input:
\begin{description}
\item
A $d$ dimensional lattice $L^d$, WLOG, $|L^d|=N^d$ (Input size to the oracle is $d\log N$ since the input size for both dimensions to the oracle is $\log N$.).
\item
An oracle function $f^d$. $f^d$ is a order preserving function. $\forall \x\in L^d,f^d(\x)\in L^d$ and $f^d(\x)\leq_c f^d(\y)$ if $\x\leq_c \y, \forall \x,\y\in L^d$.

\end{description}
\item Fixed\_point($L^d$)
\begin{enumerate}
\item If $d>1$
\begin{enumerate}
\item Let $\x^0$ be the center point in $L^d$.
\item Let $L^{d-1}=\{\x=(x_1,x_2,\cdots,x_{d-1})|(\x,x^0_d)\in L^d\}$.
\item Let $f^{d-1}(\x)=(f^{d}(\x,x^0_d)_1,f^{d}(\x,x^0_d)_2,\cdots,f^{d}(\x,x^0_d)_{d-1})$.
\item $\x^*=$Fixed\_point($L^{d-1}$).
\item If $f^d(\x^*,x^0_d)_d>x^0_d$, $L^d=\{\x|\x\ge (\x^*,x^0_d)\}$; Fixed\_point($L^d$);
\item If $f^d(\x^*,x^0_d)_d<x^0_d$, $L^d=\{\x|\x\le (\x^*,x^0_d)\}$; Fixed\_point($L^d$);
\item If $f^d(\x^*,x^0_d)_d=x^0_d$, return $(\x^*,x^0_d)$; end;
\end{enumerate}

\item If $d=1$, let $\x^L$ be the left end point and $\x^R$ be the right end point. binary\_search($\x^L,\x^R,f^d$).
\end{enumerate}

\item binary\_search($\x,\y,f$)

\begin{enumerate}
\item If $f(\x^L)=0$, output $\x^L$;
\item else if $f(\x^R)=0$, output $\x^R$;
\item else
\begin{enumerate}
\item If $f(\lfloor1/2(\x^L+\x^R)\rfloor)<\lfloor1/2(\x^L+\x^R)\rfloor$, binary\_search($\x^L,\lfloor1/2(\x^L+\x^R)\rfloor,f$);
\item If $f(\lfloor1/2(\x^L+\x^R)\rfloor)>\lfloor1/2(\x^L+\x^R)\rfloor$, binary\_search($\lfloor1/2(\x^L+\x^R)\rfloor,\x^R,f$);
\item else output $\x^*$.
\end{enumerate}
\end{enumerate}

\end{itemize}

\end{algorithm}

\begin{theorem}\label{poly-tarski}
When the order preserving function $f$ is given as an oracle, a Tarski's fixed point can be found in time $O(\log^d|L|)$ on a componentwise ordering lattice $(L,\leq_c)$.
\end{theorem}

\end{section}

\begin{section}{Determining Uniqueness under Oracle Function Model}

It has been a natural question to check whether there is another fixed point after finding the first one, such as in the applications for finding all Nash
equilibria (Echenique)\cite{e1}. In this section we develop a lower bound that,
given a general lattice $L$ with one already known fixed point, finding whether it is unique will take $\Omega(|L|)$ time for any
algorithm. Even for the componentwise ordering lattice, we also derive a $\Theta(N_1+N_2+\cdots+N_d)$ matching bounds for determining the uniqueness of the fixed point even for randomized algorithms. The technique builds on and further reveals crucial properties of mathematical structures for fixed points.

\begin{theorem}
Given a lattice $(L,\preceq)$, an order preserving function $f$ and a fixed point $\x^0$, it takes time $\Omega(|L|)$ for
any deterministic algorithm to decide whether there is a unique fixed point.
\end{theorem}

\begin{proof} Consider the lattice on a real line: $0\prec 1\prec 2\prec\cdots\prec L-1$. Let $x^0=0$, define $f(0)=0$ and $f(x)
=x-1$ for all $x\geq 1$ except a possible fixed point $x^*$.  $f(x^*)=x^*$ or $f(x^*)=x^*-1$ which is not known until we query $x^*$. Given a deterministic algorithm
$A$, define $y_j$ be the $j$-th item $A$ queried in its effort to find $x^*$. Our adversary will answer $x-1$ whenever $A$ asks for $f(x)$ until the last item when
the adversary answers $x$. Clearly this derives a lower bound of $L$.

\qed
\end{proof}

For a randomized algorithm $R$, let $p_{ij}$ be the probability $R$ queries $x=i$ on its $j$-th query. Let $k$ be the total number of queries $R$ makes.
We have:
$$\sum_{j=1}^k\sum_{i=0}^{|L|-1} p_{ij}=k.$$
Therefore, there exist $i^*$ such that
$$\sum_{j=1}^k p_{i^*j}\le\frac{k}{|L|}.$$
The adversary will place $f(i^*)=i^*$, which is queried with probability $\frac{k}{|L|}<1/2$ when we choose $k={|L|-1\over 2}$.
Therefore, we have
\begin{theorem}
Given a lattice $(L,\preceq)$, an order preserving function $f$ and a fixed point $\x^0$, it takes time $\Omega(|L|)$ for
any randomized algorithm to decide whether there is a unique fixed point with probability at least $1/2$.
\end{theorem}

As we noted before, for a lexicographic ordering lattice, it can be viewed as a total ordering lattice or componentwise ordering lattice with dimension one by an appropriate polynomial time transformation to change the oracle function for the $d$-dimension space to an oracle function on the $1$-dimension space.
Therefore,
\begin{corollary}
Given a lattice $(L,\leq_l)$, an order preserving function $f$ and given a fixed point $\x^0$, it takes time $\Omega(|L|)$ both for any deterministic algorithm and for any randomized algorithm to decide whether there is a unique fixed point with probability at least $1/2$.
\end{corollary}

Next we consider a componentwise lattice.
\begin{theorem}
Given the componentwise lattice $L=N_1\times N_2 \times\cdots\times N_d$ of $d$ dimensions, an order preserving function $f$
and a fixed point $\x^0$,
the deterministic oracle complexity is $\theta(N_1+N_2+\cdots + N_d)$ to
decide whether there is a unique fixed point.
\end{theorem}
\begin{proof}
For dimension $d\geq 2$, let $L=\{\x\in Z: \bz\leq_c \x\leq_c (N_1,N_2,\cdots,N_d)\}$.
For $\x=(x_1, x_2, \cdots, x_d)$, let $maxindex(\x)=\max\{i: x_i>0\}$ for any non-zero vector $\x$.
Define $auxi(\x) = -\e_{maxindex(\x)}$ where $\e_i$ is a unit vector in $i$-th coordinate.
Therefore, $auxi(\cdot)$ is well defined on nonzero vectors in the lattice $L$.  One example of two dimension case is demonstrated in Fig.\ref{Fig3}. The fixed point is denoted by the red color. The direction of all the other points are defined by the function $auxi(\cdot)$.
\begin{figure}[H]
\centering
\includegraphics[scale=0.4]{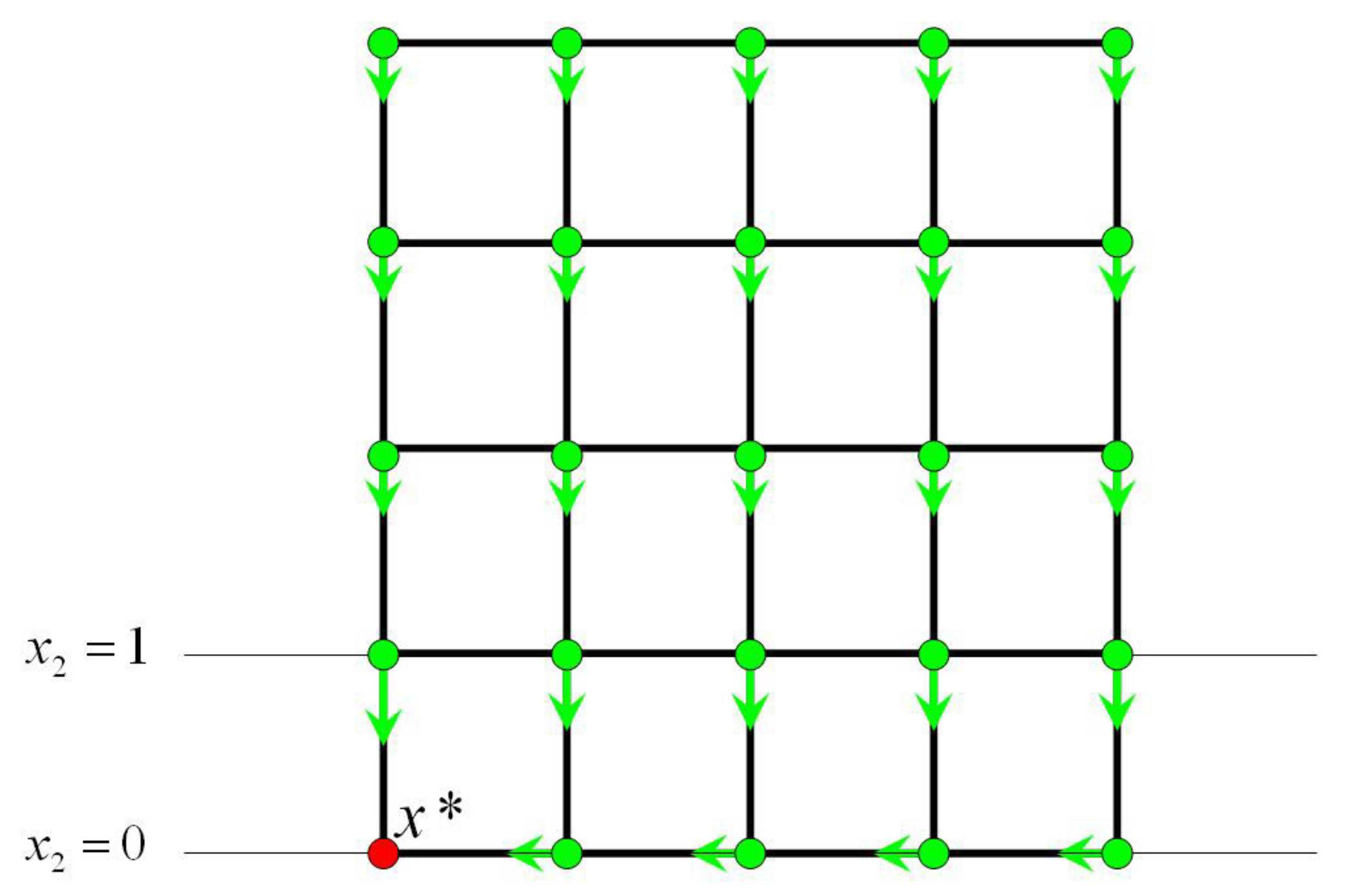}
\caption{$auxi(\x)$\label{Fig3}}
\end{figure}

The adversary will set $g(\x)=f(\x)-\x$ to be $auxi(\x)$ except at certain points (to be decided according to the algorithm) where it may hide a zero point.

\begin{enumerate}
\item Proof of the Lower Bound:

First consider $\x$ such that $x_d=0$. It constitute a solution of $d-1$ dimension. By inductive
hypothesis, it requires time $N_1+N_2+\cdots+N_{d-1}$ to decide whether or not there is one zero point
at $x_d=0$.

Second, when there is no such zero point, we need to decide if there is a zero point at
$\x$ with $x_d>0$. Fixing any $i>0$, we will set, for all $\x$ with $x_d=i$, $g(\x)=0$ whenever none of such
$\x$ is queried, and set $g(\x)=-e_d$ otherwise.
This will take $N_d$ queries.

One may note that the adversary always answers a non-zero value. In fact, for any pair
$i=maxindex(\x)$ and $j=x_i$ not query, the adversary can make $g(\x)=0$ without violating the order preserving
property.

\item Proof of the Upper Bound:

We design an algorithm which always queries the componentwise maximum point of the lattice
$\x^{\max}=(N_1,N_2,\cdots, N_d)$. We should have $g(\x^{\max})\leq_c \bz$. We are done if it is zero. Otherwise,
there must exist some $i$, such that $g(\x^{\max})_i< 0$. The problem is reduced to a smaller lattice $L'= \{\x\in Z: \bz\leq_c \x\leq_c (N_1,N_2,\cdots, N_{i-1}, N_{i}-1, N_{i+1},\cdots, N_d)\}$ which has a total sum of
side lengths at most $N_1+N_2+\cdots+N_d-1$. The claim follows.
\end{enumerate}
\qed
\end{proof}

For the randomized lower bound, it follows in the same way as in the one-dimensional case for general lattice. We can
always set $f(\x)=0$ for all $\x$ with $i=maxindex(\x)$ and $j=x_i$ if none
of such $\x$ is queried.

\begin{corollary}
Given the componentwise lattice $L=N_1\times N_2 \times\cdots\times N_d$ of $d$ dimensions, an order preserving function $f$
and a fixed point $\x^0$, it takes time $\theta(N_1+N_2+\cdots + N_d)$ for any randomized algorithm to
decide whether there is a unique fixed point with probability at least $1/2$.
\end{corollary}

\end{section}

\begin{section}{Determining Uniqueness under Polynomial Function Model}

In this section, we consider the dimension as a part of the input size in unary
and develop a hardness proof for the polynomial function model for determining the uniqueness of a given fixed point.
We start with a polynomial-time reduction from the 3-SAT problem which is NP-complete to one of finding a second Tarski's fixed point, by deriving an order preserving mapping $f$ from a componentwise ordering lattice $L$ into itself, with a given fixed point.
Therefore, given $f$ as a polynomial time
function with a known fixed point, determining whether $f$ has another fixed point in $L$ is an NP-hard problem. In other words, determining the uniqueness
of a Tarski's fixed point is co-NP-hard.

Furthermore, even for the case when the dimension is one, the uniqueness problem is still co-NP-hard. This can be done by
 designing a polynomial-time reduction from the 3-SAT problem to the uniqueness of Tarski's fixed point in a lexicographic lattice. As the lexicographic order defines a total order, it can be reduced to a one dimensional problem by finding a polynomial time algorithm for the order function calculation.
It then follows that determining the uniqueness of Tarski's fixed point in a lexicographic lattice is Co-NP hard though there exists a polynomial-time algorithm for finding one Tarski's fixed point in a lexicographic lattice in any dimension.

%\subsubsection{A Simple Reduction from 3-SAT}
We start with one of the NP-complete problems, 3-SAT, defined as follows.

\begin{definition}(3CNF-formula)
A literal is a boolean variable. A clause is several literals connected with $\vee$'s. A boolean formula is in conjuctive normal form (CNF) if it is made of clauses connected with $\wedge$'s. If every clause has exactly 3 literals, the CNF-formula is called 3CNF-formula.
\end{definition}
\begin{definition} (3-SAT Problem)

\
\ Input: \ $n$ boolean variables $x_1,x_2,\cdots, x_n$
          
             \  \  \ \ \  \ \ \ \ \ \ \  $m$ clauses $C_1,C_2,\cdots, C_m$, each consisting of three literals from
              
             \ \ \ \ \  \ \ \ \ \ \ \   the set $\{x_1,\bar{x}_1, x_2,\bar{x}_2,\cdots, x_n,\bar{x}_n\}$.
               
               Output: An assignment of $\{0,1\}$ to the boolean variables $x_1,x_2,\cdots, x_n$, 
               
               \
                \ \ \  \ \ \ \ \ \ \ \  such that the 3CNF-formula $F:C_1\wedge C_2\cdots\wedge C_m= \text{true}$, i.e., 
                
                 \ \ \ \  \ \ \ \ \ \ \ \ there is at least one true literal for every clause.

%The 3-SAT problem has input $n$ boolean variables $x_1,x_2,\cdots, x_n$ and $m$ clauses $C_1,C_2,\cdots, C_m$. Each clause $C_i$ consists three literals $y_{i1},y_{i2},y_{i3}$, where $y_{ij}\in \{x_1,\bar{x}_1, x_2,\bar{x}_2,\cdots, x_n,\bar{x}_n\}$.
%The required output is an assignment of $\{1,0\}$ to the boolean variables such that there is at least one true literal for every clause.
\end{definition}

\begin{theorem}{\cite{karp}}
3-SAT is NP-complete
\end{theorem}

For both lexicographic ordering and componentwise ordering, the Co-NP-hardness results can be derived from a reduction from 3-SAT problem.

\subsection {Proof of Co-NP-hard in Lexicographic Ordering}
\begin{corollary} \label{nphard} Given lattice $(L,\leq_l)$ and an order preserving mapping $f$ as a polynomial function,
determining that $f$ has a unique fixed point in $L$ is a Co-NP hard
problem.
\end{corollary}

\begin{proof}
Consider a 3-SAT problem with a 3CNF-formula $F(x_1,x_2,\cdots, x_n)$.
We define the function $f$ as follows:
$f(-1)=-1$ and $\forall i\geq 0$, we rewrite $i$ in binary form $i_1i_2\cdots i_n$.
Let $f(i)=i$ if $F(i_1,i_2,\cdots, i_n)=true$, and $f(i)=i-1$ otherwise.

Then $f$ is an order preserving function on the lexicographic ordering lattice $L=\{-1,0,\cdots,2^n-1\}$. If we find a fixed point $f(i^*)=i^*$ and $i^*\neq -1$ on lattice $L$, we find an assignment  $(i^*_1,i^*_2,\cdots, i^*_n)$ such that $F(i^*_1,i^*_2,\cdots, i^*_n)=true$ for the 3-SAT problem. Since 3-SAT problem is NP-hard, find a second Tarski's fixed point in lexicographic ordering lattice is NP-hard. Therefore, determining the uniqueness is Co-NP-hard.
\qed
\end{proof}

\subsection{Proof of Co-NP-hard in Componentwise Ordering }

\begin{corollary} \label{nphard1} Given lattice $(L,\leq_c)$ and an order preserving mapping $f$ as a polynomial function,
determining that $f$ has a unique fixed point in $L$ is a Co-NP hard
problem.
\end{corollary}

\begin{proof}
Again we consider a 3-SAT problem with a 3CNF-formula $F(x_1,x_2,\cdots, x_n)$. For any node $v=(v_1,v_2,\cdots,v_d)$ in a d-dimensional componnentwise ordering lattice, we define the function $f(v)$ as follows:
\begin{itemize}
\item[1)] $f(v)=f((i,i,\cdots, i))$, where  $i=\max\{v_1,v_2,\cdots, v_d\}$. 
\item[2)] $f((-1,-1,\cdots,-1))=(-1,-1,\cdots,-1)$. 
\item[3)] $\forall i\geq 0$, we rewrite $i$ in binary form $i_1i_2\cdots i_n$. $f((i,i,\cdots, i))=(i,i,\cdots, i)$ if 
$F(i_1,i_2,\cdots, i_n)=true$, and $f((i,i,\cdots, i))=(i-1,i-1,\cdots, i-1)$ otherwise.
\end{itemize}

Then $f(v)$ is an order preserving function on the componnentwise ordering lattice $L=\{(v_1,\cdots,v_d): \forall i, v_i\in\{-1,0,\cdots,2^n-1\}\}$. If we find a fixed point $f(v^*)=v^*$ and $v^*\neq (-1,-1,\cdots, -1)$ on lattice $L$, we find an assignment  $(i^*_1,i^*_2,\cdots, i^*_n)$ such that $F(i^*_1,i^*_2,\cdots, i^*_n)=true$ for the 3-SAT problem. Since 3-SAT problem is NP-hard, find a second Tarski's fixed point in componentwise ordering lattice is NP-hard. Therefore, determining the uniqueness is Co-NP-hard.
\qed
\end{proof}

\end{section}

\section{Finding Equilibria in Supermodular Games}
In previous sections, we solve the computational problems of the Tarski's fixed point. We are still interested in how to find a pure Nash equilibrium  and how to determine the uniqueness of the pure Nash in a supermodular game.  As the strong connection between the equilibrium of supermodular games and the Tarski's fixed point, the question is whether the previous results hold for supermodular games.

In this section, we develop a polynomial time algorithm to find a pure Nash equilibirum which is more efficient than the algorithm we design for Tarski's fixed point before.  %Besides the time complexity for determining the uniqueness of equilibrium in supermodular games is less than determining the uniqueness of Tarski's fixed point in the oracle function model.  
But the Co-NP-hardness result still holds for supermodular games.

\subsection {Supermodular Games and Tarski's Fixed Points}

We will start with the formal definition of supermodular games.
\begin{definition} (Supermodular Games)\label{super}
Let $\Gamma=\{(S_i,u_i):i=1,\cdots,d\}$ be a finite supermodular game with $d$ players:
\begin{itemize}
\item $S_i$ is a finite subset of $\mathbf{R}$;
\item $u_i$ has increasing differences in $(s_i,s_{-i})$, where $s_{-i}$ is the strategy set of all other players except player $i$. I.e.,
$$u(s_i',s_{-i}')-u(s_i,s_{-i}')\ge u(s_i',s_{-i})-u(s_i,s_{-i}),\forall s_i'\ge s_i,s_{-i}'\ge s_{-i}$$
\end{itemize}
\end{definition}

In the following discussion, W.O.L.G, we assume $S_i=\{0,1,\cdots,N_i-1\}.$ The model can be viewed as a discretized version of a game with continuous strategy spaces, where each $S_i$ is an interval. Then $S=\times_{i=1}^dS_i$ is a componentwise lattice (see Example 1).

\begin{example} (Supermodular game and componentwise lattice)
Consider a supermodular game of two players: $S_1=\{0,1,2\}, S_2=\{0,1\}$. Then $S=S_1\times S_2$ is a componentwise ordering lattice as shown in {\bf Fig.} \ref{example of component order lattice}.
\begin{figure}[H]
\centering
{\includegraphics[scale=0.2]{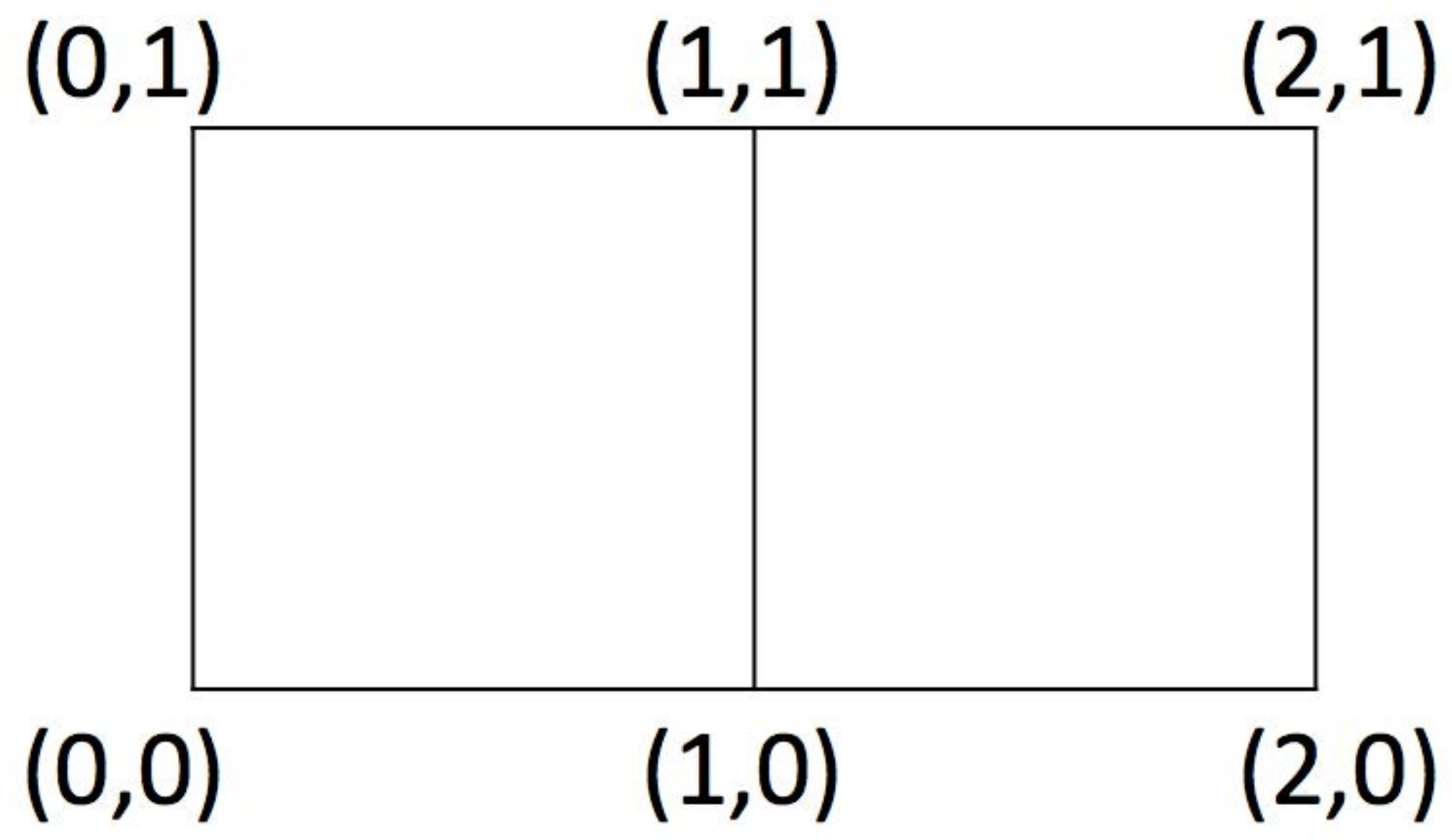}}
\caption{An example of componentwise ordering lattice\label{example of component order lattice}}
\end{figure}
\end{example}

Let $B_i$ denote $i$'s best-response function in $\Gamma$. $$B_i(s_{-i}):=arg\max_{{s_i}\in S_i}u_i({s_i},s_{-i}).$$

Denote $\overline{B}_i(s_{-i})$ as the greatest element  and $\underline{B}_i(s_{-i})$ as the least element in $B_i(s_{-i})$.

In supermodular games, by Topkis' theorem\cite{t2}, we have,
$$\overline{B}_i(s_{-i})\ge \overline{B}_i(s'_{-i}) \text{ and } \underline{B}_i(s_{-i})\ge \underline{B}_i(s'_{-i}) \text{ if }s_{-i}\ge s'_{-i}.$$

Let $\underline{B}(s)=\{\underline{B}_i(s_{-i}): i=1,\cdots,d\}$ be the least best-response function of the game, then $\underline{B}: S\rightarrow S$ is order preserving.

%In supermodular games, we have,
%\begin{itemize}
%\item $\bar{B_i}(s_{-i})\ge \bar{B_i}(s'_{-i})$, if $s_{-i}\ge s'_{-i}$.
%\item $B: S\rightarrow S$ is order preserving.
%\end{itemize}

%The existence of a pure Nash equilibrium is guaranteed by the Tarski's fixed point theorem.

%\begin{theorem}(Tarski's Fixed Point Theorem)\cite{tarski1}.
%If $(L, \preceq)$ is a complete lattice and $f$ is order preserving
%from $L$ into itself, then there exists some $\x^*\in L$ such that
%$$f(\x^*)=\x^*.$$
%\end{theorem}
Tarski's fixed point theorem guarantees the existence of fixed points of any order preserving
function $f:L\rightarrow L$ on any nonempty complete lattice. A supermodular game $(S,\preceq)$ is a complete lattice and  the least best-response function $\underline{B}$ is order-preserving from $S$ to itself. Therefore, there exists an equilibrium point $x^*\in S$ such that $\underline{B}(x^*)=x^*$.

\subsection {Equilibrium Computation in Supermodular Games}

 Recall $\underline{B}(s)=\{\underline{B}_i(s_{-i}): i=1,\cdots,d\}$, where $s\in S$ and $S=\times_{i=1}^dS_i$ is the best-response function of the supermodular game. We assume the strategy set for each player $i$ is $S_i=\{0,1,\cdots, N_i-1\}$. In Tarski's fixed point theorem, the only requirement for function $f$ is order-preserving. In supermodular game, $\underline{B}$ is not only order-preserving but also need to be consistency, since for the same $s_{-i}$, the value of $\underline{B}_i(s_{-i})$ should be the same.
Therefore. if there exists an algorithm $A$ that finds a Tarski's fixed point for any componentwise lattice with order-preserving function $f$ in time T, $A$ can finds an equilibrium in supermodular game in time $T$. However, not vice versa.

W.L.O.G., assume $N_1\le N_2\le\cdots\le N_d$.

\begin{theorem}
When the best response function $\underline{B}$ is given as an oracle, a pure Nash equilibrium can be found in time $O(\log N_1\log N_2\cdots\log N_{d-1})$ in a supermodular game $\Gamma$.
\end{theorem}

\begin{proof} 
The algorithm is similar to the proof of Theorem \ref{poly-tarski} for finding one Tarski's fixed point. The only difference here is for the 2D case. On a $N_1\times N_2$ box, we start with the node $(\lfloor\frac{N_1}{2}\rfloor,y)$, where $1\le y\le N_2$ can be any integer. We query the value of $\underline{B}(\lfloor\frac{N_1}{2}\rfloor,y)$. Assume the value is $(x',y')$. Next we query the value of $\underline{B}(\lfloor\frac{N_1}{2}\rfloor,y')$. Because in the previous query we have already known that $\underline{B}_2(\lfloor\frac{N_1}{2}\rfloor)=y'$, we must have $\underline{B}(\lfloor\frac{N_1}{2}\rfloor,y')=(x'',y')$.
\begin{enumerate}
\item If $x''=\lfloor\frac{N_1}{2}\rfloor$, $(\lfloor\frac{N_1}{2}\rfloor,y')$ is a pure Nash equilibrium.
\item If $x''<\lfloor\frac{N_1}{2}\rfloor$, all nodes smaller than $(x'',y')$ form a complete lattice and the size is less than half of the original lattice.
\item If $x''>\lfloor\frac{N_1}{2}\rfloor$, all nodes greater than $(x'',y')$ form a complete lattice and the size is also less than half of the original lattice.
\end{enumerate}
Therefore, by using the property of the best response function, a pure Nash can be found in time $2\log N_1$ for 2D case. Recall that finding a Tarski's fixed point takes time $O(\log N_1\log N_2)$ for 2D case. The generalization for the higher dimensional cases is similar to what we do for finding one Tarski's fixed point. Thus we can find a pure Nash in time $O(\log N_1\log N_2\cdots\log N_{d-1})$.
\qed

\end{proof}

%To determine the uniqueness of equilibrium in the oracle function model, we have
%
%\begin{theorem}
%Given a $d$-player supermodular game with strategy set $S=\times_{i=1}^dS_i: \forall i, |S_i|=N_i$ and a best response function $B$ as a oracle function
%as well as a given pure Nash equilibrium $\x^0$.
%The matching oracle query complexity is $\theta(N_1+N_2+\cdots + N_{d-1})$ to
%determine whether there is a unique equilibrium.
%\end{theorem}

%\begin{proof}
%The idea is similar to the proof of Theorem \ref{uni-com}.
%\end{proof}

Again, by a reduction from 3-SAT problem, we obtain

\begin{theorem}
Given a $d$-player supermodular game $\Gamma$ with strategy set $S=\times_{i=1}^dS_i: \forall i, S_i=\{-1,0,\cdots, 2^n-1\}$ and a best response function $B$ %$\underline{B}$ 
as a polynomial function determining that $\Gamma$ has a unique Nash equilibrium is a Co-NP hard
problem.
\end{theorem}
\begin{proof}
%Similar to the prove for Tarski's fixed point on componentwise ordering lattice, 
Consider a 3-SAT problem with a 3CNF-formula $F(x_1,x_2,\cdots, x_n)$. ${B}(s)=\{B_i(s_{-i}): i=1,\cdots,d\}$, where $s\in S$ and $S=\times_{i=1}^dS_i$ is the best response function of the supermodular game. 

We define $\forall i, {B}_i(s_{-i})$ as follows:
\begin{itemize}
\item[1)] ${B}_i(s_{-i})={B}_i((j,j,\cdots, j))$, where  $j=\max\{s_{-i}\}$. 
\item[2)] ${B}_i((-1,-1,\cdots,-1))=-1$. 
\item[3)] $\forall j\geq 0$, we rewrite $j$ in binary form $j_1j_2\cdots j_n$. ${B}_i((j,j,\cdots, j))=j$ if 
$F(j_1,j_2,\cdots, j_n)=true$, and ${B}_i((j,j,\cdots, j))=j-1$ otherwise.
\end{itemize}

For $\forall s=\{s_1,s_2,\cdots, s_d\}>\{s_1',s_2',\cdots,s_d'\}=s'$, we have $B(s)\ge B(s')$, which implies B(s) is an order preserving function.  Then ${B}(s)=\{{B}_i(s_{-i}): i=1,\cdots,d\}$ is a best response function for the supermodular game $\Gamma$. 

Next we prove that if $s^*=(s_1^*,s_2^*,\cdots, s_d^*)$ is an equilibrium of the supermodular game $\Gamma$ with the above best response function $B$, we must have $s_1^*=s_2^*=\cdots=s_d^*$. I.e., all the elements of $s^*$ must be identical.  

Suppose $(x_1,x_2,\cdots,x_d)$ is an equilibrium.  By the definition of $B_i$, $$\max(x_2,x_3,\cdots, x_d)=x_1 \text{ or } x_1+1.$$
\begin{itemize}
\item[1)] Case I: $\max(x_2,x_3,\cdots, x_d)=x_1$. Then $\max(x_1,x_3,\cdots, x_d)=x_1$, so $B_{2}(x_{-2})=B_{2}(x_1,x_1,\cdots, x_1)=x_1$. Since $x$ is an equilibrium, $x_2=B_{2}(x_{-2})$. Therefore, $x_2=x_1$. Similarly, we can prove $\forall i, x_i=x_1$. 
\item[2)] Case 2: $\max(x_2,x_3,\cdots, x_d)=x_1+1$. We consider two cases, a)only one element of $\{x_2,x_3,\cdots, x_d\}$ is $x_1+1$  and b) at least two elements equal to $x_1+1$. 

\begin{itemize}
\item[a)] W.L.O.G, assume only $x_2=x_1+1$. Then $$B_{2}(x_{-2})=B_{2}(x_1,x_1,\cdots, x_1)=x_1\text{ or }x_1-1<x_2.$$ This contradicts to the assumption that $x$ is an equilibrium.
\item[b)] W.L.O.G, assume $x_2=x_3=x_1+1$. Then $$B_{2}(x_{-2})=B_{2}(x_1+1,x_1+1,\cdots, x_1+1)=B_{1}(x_{-1})=x_1<x_2.$$ Again it contradicts to the assumption that $x$ is an equilibrium.

\end{itemize}
\end{itemize}
Therefore, $s^*=(s_1^*,s_2^*,\cdots, s_d^*)$ is an equilibrium of the supermodular game $\Gamma$ with the above best response function $B$, we must have $s_1^*=s_2^*=\cdots=s_d^*$.
Hence, if we find an equilibrium $s^*$  and $s^*\neq (-1,-1,\cdots, -1)$ in $\Gamma$, we find an assignment  $(j^*_1,j^*_2,\cdots, j^*_n)$ such that $F(j^*_1,j^*_2,\cdots, j^*_n)=true$ for the 3-SAT problem. Since 3-SAT problem is NP-hard, find a second equilibrium is NP-hard. Therefore, determining the uniqueness is Co-NP-hard.

\qed
\end{proof}

\section{Conclusion and Open Problems}
Results on the Tarski's fixed points contrast with past results for the general fixed point computation in several ways. First in the oracle function model, several fixed point computational problems are known to be require an exponential number of queries for constant dimensions, including the two dimensional case (Chen and Deng; Hirsch et al.)\cite{CD2005,HPV}. Our results prove the Tarski's fixed point to be polynomial in the oracle model. It also follows that it is so for the polynomial function model, which is also different for those fixed point computational problems which are known to be PPAD-complete for constant dimensions, including the two dimensional case (Chen and Deng)\cite{ChenDeng2006}.

Recently, Mihalis, Kusha and Papadimitriou stated in a private communication that they proved a lower bound of $\Omega(\log^2|L|)$ in the oracle function model for finding a Tarski's fixed point in the two dimensional case. Together with our upper bound results, we conjecture a matching bound of  $\Theta(\log^d|L|)$ for general $d$. 

In the polynomial function model, we prove that determining the uniqueness is co-NP-complete. In comparison, the uniqueness of Nash equilibrium is known to be co-NP-complete but its existence is in PPAD.

The above comparisons with previous work leave the following outstanding open problem: Is it PPAD-complete to find a Tarski's fixed point in the variable dimension $n$ for the polynomial function model?
This problem is known to be true for finding a Sperner simplex in dimension $n$ when $n$ is a variable.
We conjecture that this is also true for finding a Tarski's fixed point.

%\begin{section}{Conclusion and Open Problems}
%Results on the Tarski's fixed points contrast with past results for the general fixed point computation in several ways. First in the oracle function model, several fixed point computational problems are known to be require an exponential number of queries for constant dimensions, including the two dimensional case (Chen et al., Deng et al. and Hirsch et al.)\cite{CD2005,dqsz,HPV}. Our results prove the Tarski's fixed point to be polynomial in the oracle model. It also follows that it is so for the polynomial function model, which is also different for those fixed point computational problems which are known to be PPAD-complete for constant dimensions, including the two dimensional case (Chen et al. and  Deng et al.)\cite{ChenDeng2006, dqsz}. In the polynomial function model, we prove that determining the uniqueness is co-NP-complete. In comparison, the uniqueness of Nash equilibrium is known to be co-NP-complete but its existence is in PPAD.
%
%The above comparisons with previous work leave the following outstanding open problem: Is it PPAD-complete to find a Tarski's fixed point in the variable dimension $n$ for the polynomial function model ?
%This problem is known to be true for finding a Sperner simplex in dimension $n$ with $n$ as a variable.
%We conjecture that this is also true for finding a Tarski's fixed point.
%
%\end{section}
%

\section{Appendix: Alternative proofs for Co-NP-hardness}

Let \[P=\{\x\in R^n\;|\;A\x\leq \bb\}\]  be a full-dimensional polytope,
where $A$ is an $m\times n$ rational matrix satisfying that each row
of $A$ has at most one positive entry and $\bb$ a rational vector of
$R^m$. It has been shown in (Lagarias)\cite{lagarias}  that
\begin{theorem}\cite{lagarias}\label{npcomplete} Determining whether there is an integer point
in $P$ is an NP-complete problem.
\end{theorem}

\subsection {Proof of Co-NP-hard in Lexicographic Ordering}
\begin{corollary} \label{nphard} Given lattice $(L,\leq_l)$ and an order preserving mapping $f$ as a polynomial function,
determining that $f$ has a unique fixed point in $L$ is a Co-NP hard
problem.
\end{corollary}

We assume $n\geq 2$. Similarly, let
 $\x^{\max}=(x^{\max}_1,x^{\max}_2,\ldots,x_n^{\max})^{\top}$ with
$x^{\max}_j=\max_{\x\in P}x_j$, $j=1,2,\ldots,n$, and
$\x^{\min}=(x^{\min}_1,x^{\min}_2,\ldots,x^{\min}_n)^{\top}$ with
$\x^{\min}_j=\min_{\x\in P}x_j$, $j=1,2,\ldots,n$. Let $D(P)=\{\x\in Z^n\;|\;x^l\leq_l
\x\leq_l \x^u\}$, where $\x^u=\lfloor \x^{\max}\rfloor$ and $\x^l=\lfloor
\x^{\min}\rfloor$.

For $\y\in R^n$ and $k\in N\cup\{0\}$,
 let
\[P(\y,k)=\left\{\begin{array}{ll}
P & \mbox{if $k=0$,}\\
 \{\x\in P\;|\;x_i=y_i,\;i=1,2,\ldots,k\} & \mbox{otherwise.}
 \end{array}\right.\]
\begin{definition}\label{imd} For $\y\in D(P)$,
$h(\y)=(h_1(\y),h_2(\y),\ldots,h_n(\y))^{\top}\in D(P)$ is given as
follows:
\begin{description}
\item[Step 1:]
 If $y_1=x_1^l$, let $h(\y)=x^l$. If $\y\in P$, let $h(\y)=\y$. Otherwise, let $k=2$ and go to {\bf Step
 2}.
\item[Step 2:] Solve the linear program
\[\begin{array}{rl}
\min & x_k-v_k\\
\mbox{subject to} & \x\in P(\y,k-1)\mbox{ and }\bv\in P(\y,k-1),
\end{array}\]
to obtain its optimal solution $(\x^*,\bv^*)$. Let
$d_k^{\min}(\y)=x_k^*\mbox{ and } d^{\max}_k(\y)=v_k^*$. If
$y_k\geq\lceil d_k^{\min}(\y)\rceil$, go to {\bf Step 3}. Otherwise,
go to {\bf Step 4}.
\item[Step 3:] If $\lfloor d^{\max}_k(\y)\rfloor<\lceil
d_k^{\min}(\y)\rceil$, go to {\bf Step 4}. Otherwise, go to {\bf Step
5}.
 \item[Step 4:]Let $p(\y)=k$. If
$y_{k-1}\leq x^l_{k-1}+1$, let \[h_i(y)=\left\{\begin{array}{ll}
y_i & \mbox{if $1\leq i\leq k-2$,}\\
 x^l_i & \mbox{if $k-1\leq i\leq n$,}
\end{array}\right.\]
$i=1,2,\ldots,n$. Otherwise, let
\[h_i(\y)=\left\{\begin{array}{ll}
y_i & \mbox{if $1\leq i\leq k-2$,}\\
y_{k-1}-1 & \mbox{if $i=k-1$,}\\
 x_i^u & \mbox{if $k\leq i\leq n$,}
 \end{array}\right.\]
$i=1,2,\ldots,n$.
\item[Step 5:] If $y_k>\lfloor d^{\max}_k(\y)\rfloor$, let $p(\y)=k$ and
\[h_i(\y)=\left\{\begin{array}{ll}
y_i & \mbox{if $1\leq i\leq k-1$,}\\
 \lfloor d^{\max}_k(\y)\rfloor & \mbox{if $i=k$,}\\
x^u_i & \mbox{if $k+1\leq i\leq n$,}
\end{array}\right.\]
$i=1,2,\ldots,n$. Otherwise, let $k=k+1$ and go to {\bf Step 2}.
\end{description}
\end{definition}
\begin{lemma}
\label{lessy}
 $\x^l\leq h(\y)\leq_l \y$ and $h(\y)\neq \y$ for all $\y\in D(P)$ with
 $\y\neq \x^l$ and $\y\notin P$.
\end{lemma}
{\bf Proof.} Clearly, the lemma holds for all $\y\in D(P)$ with
$y_1=x_1^l$ and $\y\neq \x^l$. Let $\y$ be any given point in $D(P)$
with $y_1\neq x_1^l$ and $\y\notin P$ and $k=p(\y)$. From the
definition of $h(\y)$, we obtain that $k$ is well defined, $k\geq 2$,
and $x_i^l<\lceil d_k^{\min}(\y)\rceil\leq y_i\leq\lfloor
d^{\max}_k(\y)\rfloor$, $i=1,2,\ldots,k-1$. Furthermore, one of the
following five cases must occur.
\begin{description}
\item[Case 1:] $y_k\geq\lceil
d_k^{\min}(\y)\rceil$, $\lfloor d^{\max}_k(\y)\rfloor<\lceil
d_k^{\min}(\y)\rceil$ and  $y_{k-1}\leq x^l_{k-1}+1$. From {\bf Step
4}, we find that
\[h_i(\y)=\left\{\begin{array}{ll}
y_i & \mbox{if $1\leq i\leq k-2$,}\\
 x^l_i & \mbox{if $k-1\leq i\leq n$,}
\end{array}\right.\]
$i=1,2,\ldots,n$. Thus, it follows from $y_{k-1}>x_{k-1}^l$ that
$\x^l\leq h(\y)\leq_l\y$ and $h(\y)\neq \y$.
\item[Case 2:]  $y_k\geq\lceil
d_k^{\min}(\y)\rceil$, $\lfloor d^{\max}_k(\y)\rfloor<\lceil
d_k^{\min}(\y)\rceil$ and  $y_{k-1}> x^l_{k-1}+1$. From {\bf Step 4},
we find that
\[h_i(\y)=\left\{\begin{array}{ll}
y_i & \mbox{if $1\leq i\leq k-2$,}\\
y_{k-1}-1 & \mbox{if $i=k-1$,}\\
 x_i^u & \mbox{if $k\leq i\leq n$,}
 \end{array}\right.\]
$i=1,2,\ldots,n$. Thus, it follows from $y_{k-1}-1<y_{k-1}$ that
$\x^l\leq h(\y)\leq_l\y$ and $h(\y)\neq \y$.
\item[Case 3:] $y_k\geq\lceil
d_k^{\min}(\y)\rceil$, $\lfloor d^{\max}_k(\y)\rfloor\geq\lceil
d_k^{\min}(\y)\rceil$ and $y_k>\lfloor d^{\max}_k(\y)\rfloor$. From
{\bf Step 5}, we find that
\[h_i(\y)=\left\{\begin{array}{ll}
y_i & \mbox{if $1\leq i\leq k-1$,}\\
 \lfloor d^{\max}_k(\y)\rfloor & \mbox{if $i=k$,}\\
x^u_i & \mbox{if $k+1\leq i\leq n$,}
\end{array}\right.\]
$i=1,2,\ldots,n$. Thus, it follows from $y_k>\lfloor
d^{\max}_k(\y)\rfloor$ that $\x^l\leq h(\y)\leq_l\y$ and $h(\y)\neq \y$.
\item[Case 4:]$y_k<\lceil
d_k^{\min}(\y)\rceil$ and  $y_{k-1}\leq x^l_{k-1}+1$. From {\bf Step
4}, we find that
\[h_i(\y)=\left\{\begin{array}{ll}
y_i & \mbox{if $1\leq i\leq k-2$,}\\
 x^l_i & \mbox{if $k-1\leq i\leq n$,}
\end{array}\right.\]
$i=1,2,\ldots,n$. Thus, it follows from $y_{k-1}>x_{k-1}^l$ that
$\x^l\leq h(\y)\leq_l\y$ and $h(\y)\neq \y$.
\item[Case 5:]Consider the case that $y_k<\lceil
d_k^{\min}(\y)\rceil$ and  $y_{k-1}> x^l_{k-1}+1$. From {\bf Step 4},
we find that \[h_i(\y)=\left\{\begin{array}{ll}
y_i & \mbox{if $1\leq i\leq k-2$,}\\
y_{k-1}-1 & \mbox{if $i=k-1$,}\\
 x_i^u & \mbox{if $k\leq i\leq n$,}
 \end{array}\right.\]
$i=1,2,\ldots,n$. Thus, it follows from $y_{k-1}-1<y_{k-1}$ that
$\x^l\leq h(\y)\leq_l\y$ and $h(\y)\neq \y$.
\end{description}
Therefore, it always holds that $\x^l\leq h(\y)\leq_l\y$ and $h(\y)\neq
\y$. The proof is completed.\hfill\fbox{}

 As a corollary of
Lemma~\ref{lessy}, we obtain that
\begin{corollary} For any given $\x^*\in D(P)$, $\x^*\in P$ if and only if
$h(\x^*)=\x^*$ and $\x^*\neq \x^l$.
\end{corollary}

\begin{theorem}
\label{im} Under the lexicographic ordering, $h$ is an order preserving
mapping from $D(P)$ into itself.
\end{theorem}
{\bf Proof.} Let $\y^1$ and $\y^2$ be any given two points in $D(P)$
with $\y^1\leq_l\y^2$ and $\y^1\neq \y^2$. Let $q$ be the index in $N$
satisfying that $y^1_i=y^2_i$, $i=1,2,\ldots,q-1$, and
$y^1_q<y^2_q$. From the definition of $h$, we obtain that
$h(\y^1)=\x^l$  if
 $y^1_1=x_1^l$ and that $h(\y^2)=\y^2$ if $\y^2\in P$. Thus, when $y^1_1=x_1^l$ or $\y^2\in P$, it follows from Lemma~\ref{lessy} that
$h(\y^1)\leq_lh(\y^2)$.

Suppose that $y^1_1\neq x_1^l$ and $\y^2\notin P$. Let $k_1=p(\y^1)$
and $k_2=p(\y^2)$. From the definition of $h(\y^2)$, we obtain that
$k_2$ is well defined and $k_2\geq 2$.
\begin{description}
\item[Case 1:] $2\leq k_2\leq q-1$. From $y_i^1=y^2_i$,
$i=1,2,\ldots,q-1$, we derive that $k_1=k_2$. Thus, $h(\y^1)=h(\y^2)$.
Therefore, $h(\y^1)\leq_lh(\y^2)$.
\item[Case 2:] $2\leq k_2=q$. From the definition of $k_2=p(\y^2)$, we know that
\[x_i^l<\lceil
d_{i}^{\min}(\y^2)\rceil\leq y^2_i\leq\lfloor
d^{\max}_{i}(\y^2)\rfloor,\;i=1,2,\ldots,q-1.\] Since $y^1_i=y^2_i$,
$i=1,2,\ldots,q-1$, hence, $\lceil d_{i}^{\min}(\y^1)\rceil=\lceil
d_{i}^{\min}(\y^2)\rceil$ and $\lfloor
d^{\max}_{i}(\y^1)\rfloor=\lfloor d^{\max}_{i}(\y^2)\rfloor$,
$i=1,2,\ldots,q$, and $x_i^l<\lceil d_{i}^{\min}(\y^1)\rceil\leq
y^1_i\leq\lfloor d^{\max}_{i}(\y^1)\rfloor$, $i=1,2,\ldots,q-1$.
\begin{enumerate}
\item Suppose that $y^2_{q}\geq\lceil
d_{q}^{\min}(\y^2)\rceil$, $\lfloor d^{\max}_{q}(y^2)\rfloor<\lceil
d_{q}^{\min}(\y^2)\rceil$ and $y^2_{q-1}\leq x^l_{q-1}+1$. From {\bf
Step 4}, we find that
\[h_i(\y^2)=\left\{\begin{array}{ll}
y^2_i & \mbox{if $1\leq i\leq q-2$,}\\
 x^l_i & \mbox{if $q-1\leq i\leq n$,}
\end{array}\right.\]
$i=1,2,\ldots,n$. Since $\lfloor d^{\max}_{q}(\y^2)\rfloor<\lceil
d_{q}^{\min}(\y^2)\rceil$, $\lceil d_{q}^{\min}(\y^1)\rceil=\lceil
d_{q}^{\min}(\y^2)\rceil$ and $\lfloor
d^{\max}_{q}(\y^1)\rfloor=\lfloor d^{\max}_{q}(\y^2)\rfloor$, we
derive that $k_1=k_2=q$. Thus, it follows from
$y^1_{q-1}=y^2_{q-1}\leq x^l_{q-1}+1$ and {\bf Step 4} that
\[h_i(\y^1)=\left\{\begin{array}{ll}
y^1_i & \mbox{if $1\leq i\leq q-2$,}\\
 x^l_i & \mbox{if $q-1\leq i\leq n$,}
\end{array}\right.\]
$i=1,2,\ldots,n$. Therefore,
 $h(\y^1)=h(\y^2)$, and consequently $h(\y^1)\leq_l h(\y^2)$.

 \item Suppose that
$y^2_{q}\geq\lceil d_{q}^{\min}(\y^2)\rceil$, $\lfloor
d^{\max}_{q}(\y^2)\rfloor<\lceil d_{q}^{\min}(\y^2)\rceil$ and
$y^2_{q-1}> x^l_{q-1}+1$. From {\bf Step 4}, we find that
\[h_i(\y^2)=\left\{\begin{array}{ll}
y^2_i & \mbox{if $1\leq i\leq q-2$,}\\
y^2_{q-1}-1 & \mbox{if $i=q-1$,}\\
 x_i^u & \mbox{if $q\leq i\leq n$,}
 \end{array}\right.\]
$i=1,2,\ldots,n$. Since $\lfloor d^{\max}_{q}(\y^2)\rfloor<\lceil
d_{q}^{\min}(\y^2)\rceil$, $\lceil d_{q}^{\min}(\y^1)\rceil=\lceil
d_{q}^{\min}(\y^2)\rceil$ and $\lfloor
d^{\max}_{q}(\y^1)\rfloor=\lfloor d^{\max}_{q}(\y^2)\rfloor$, we
derive that $k_1=k_2=q$. Thus, it follows from
$y^1_{q-1}=y^2_{q-1}>x^l_{q-1}+1$ and {\bf Step 4} that
\[h_i(\y^1)=\left\{\begin{array}{ll}
y^1_i & \mbox{if $1\leq i\leq q-2$,}\\
y^1_{q-1}-1 & \mbox{if $i=q-1$,}\\
 x_i^u & \mbox{if $q\leq i\leq n$,}
 \end{array}\right.\]
$i=1,2,\ldots,n$. Therefore,
 $h(\y^1)=h(\y^2)$, and consequently $h(\y^1)\leq_l h(\y^2)$.

\item Suppose that $y^2_{q}\geq\lceil
d_{q}^{\min}(\y^2)\rceil$, $\lfloor
d^{\max}_{q}(\y^2)\rfloor\geq\lceil d_{q}^{\min}(\y^2)\rceil$ and
$y^2_{q}>\lfloor d^{\max}_{q}(\y^2)\rfloor$. From {\bf Step 5}, we
find that
\[h_i(\y^2)=\left\{\begin{array}{ll}
y^2_i & \mbox{if $1\leq i\leq q-1$,}\\
 \lfloor d^{\max}_{q}(\y^2)\rfloor & \mbox{if $i=q$,}\\
x^u_i & \mbox{if $q+1\leq i\leq n$,}
\end{array}\right.\]
$i=1,2,\ldots,n$.
\begin{itemize}
\item
Consider that $\y^1\in P$. Then, $h(\y^1)=\y^1$ and $\lceil
d_{i}^{\min}(\y^1)\rceil\leq y^1_{i}\leq \lfloor
d^{\max}_{i}(\y^1)\rfloor$, $i=1,2,\ldots,n$.
 Thus, from $\lfloor d^{\max}_{q}(\y^1)\rfloor=\lfloor
 d^{\max}_{q}(\y^2)\rfloor$, we obtain that $h_{q}(\y^1)=y^1_q\leq \lfloor
d^{\max}_{q}(\y^2)\rfloor=h_{q}(\y^2)$. Therefore,
 $h(\y^1)\leq_lh(\y^2)$ follows from $h_i(\y^1)=h_i(\y^2)$, $i=1,2,\ldots,q-1$, and $h_i(\y^1)\leq h_i(\y^2)$, $i=q,q+1,\ldots,n$.

\item
Consider that $\y^1\notin P$. From $y^1_i=y^2_i$, $i=1,2,\ldots,q-1$,
and $k_2=q$, we derive that $k_1\geq q$.
\begin{enumerate}
\item Assume that $k_1=q$. Since $\lfloor
d^{\max}_{q}(\y^1)\rfloor\geq\lceil d_{q}^{\min}(\y^1)\rceil$, hence,
either $y^1_q>\lfloor d^{\max}_{q}(\y^1)\rfloor$ or $y^1_q<\lceil
d_{q}^{\min}(\y^1)\rceil$. Thus, from the definition of $h(\y^1)$, we
obtain that, when $y^1_q>\lfloor d^{\max}_{q}(\y^1)\rfloor$,
\[h_i(\y^1)=\left\{\begin{array}{ll}
y^1_i & \mbox{if $1\leq i\leq q-1$,}\\
 \lfloor d^{\max}_{q}(\y^1)\rfloor & \mbox{if $i=q$,}\\
x^u_i & \mbox{if $q+1\leq i\leq n$,}
\end{array}\right.\]
$i=1,2,\ldots,n$; and when $y^1_q<\lceil d_{q}^{\min}(\y^1)\rceil$,
if $y^1_{q-1}\leq x_{q-1}^l+1$,
\[h_{i}(\y^1)=\left\{\begin{array}{ll}
y^1_i & \mbox{if $1\leq i\leq q-2$,}\\
x^l_i & \mbox{otheriwse,}
\end{array}\right.\]
$i=1,2,\ldots,n$, and if $y^1_{q-1}> x_{q-1}^l+1$,
\[h_i(\y^1)=\left\{\begin{array}{ll}
y^1_i & \mbox{if $1\leq i\leq q-2$,}\\
y^1_{q-1}-1 & \mbox{if $i=q-1$,}\\
 x_i^u & \mbox{if $q\leq i\leq n$,}
 \end{array}\right.\]
$i=1,2,\ldots,n$. Therefore, when $y^1_q>\lfloor
d^{\max}_{q}(\y^1)\rfloor$, $h(\y^1)\leq_lh(\y^2)$ follows from
$h(\y^1)=h(\y^2)$; and when $y^1_q<\lceil d_{q}^{\min}(\y^1)\rceil$, if
$y^1_{q-1}\leq x_{q-1}^l+1$, then $h(\y^1)\leq_lh(\y^2)$ follows from
$h_i(\y^1)=h_i(\y^2)$, $i=1,2,\ldots,q-2$, and
$h_{q-1}(\y^1)=x_{q-1}^l<y^2_{q-1}=h_{q-1}(\y^2)$, and if $y^1_{q-1}>
x_{q-1}^l+1$, then $h(\y^1)\leq_lh(\y^2)$ follows from
$h_i(\y^1)=h_i(\y^2)$, $i=1,2,\ldots,q-2$, and
$h_{q-1}(\y^1)=y^1_{q-1}-1<y^1_{q-1}=y^2_{q-1}=h_{q-1}(\y^2)$.

\item Assume that $k_1>q$. Then, $k_1-1\geq q$ and
$\lceil d_{i}^{\min}(\y^1)\rceil\leq y^1_{i}\leq \lfloor
d^{\max}_{i}(\y^1)\rfloor$, $i=1,2,\ldots,k_1-1$.
 Thus, from the definition of $h(\y^1)$, we obtain that
$h_i(\y^1)=y^1_i=y_i^2=h_i(\y^2)$, $i=1,2,\ldots,q-1$,  $h_q(\y^1)\leq
y^1_q\leq \lfloor d^{\max}_{q}(\y^1)\rfloor=\lfloor
d^{\max}_{q}(\y^2)\rfloor=h_q(\y^2)$, and $h_i(\y^1)\leq
x_i^u=h_i(\y^2)$, $i=q+1,q+2,\ldots,n$. Therefore,
$h(\y^1)\leq_lh(\y^2)$.
\end{enumerate}
\end{itemize}

\item Suppose that $y^2_{q}<\lceil
d_{q}^{\min}(\y^2)\rceil$ and  $y^2_{q-1}\leq x^l_{q-1}+1$. From {\bf
Step 4}, we find that
\[h_i(\y^2)=\left\{\begin{array}{ll}
y^2_i & \mbox{if $1\leq i\leq q-2$,}\\
 x^l_i & \mbox{if $q-1\leq i\leq n$,}
\end{array}\right.\]
$i=1,2,\ldots,n$. Since $y^1_i=y^2_i$, $i=1,2,\ldots,q-1$,
$y^1_{q}<y^2_q$, $\lceil d_{q}^{\min}(\y^1)\rceil=\lceil
d_{q}^{\min}(\y^2)\rceil$, and $k_2=q$, hence, we derive that
$k_1=k_2=q$ and $y^1_{q}<\lceil d_{q}^{\min}(\y^1)\rceil$. Thus, it
follows from $y^1_{q-1}=y^2_{q-1}$ and {\bf Step 4} that
\[h_i(\y^1)=\left\{\begin{array}{ll}
y^1_i & \mbox{if $1\leq i\leq q-2$,}\\
 x^l_i & \mbox{if $q-1\leq i\leq n$,}
\end{array}\right.\]
$i=1,2,\ldots,n$. Therefore, $h(\y^1)=h(\y^2)$ and
$h(\y^1)\leq_lh(\y^2)$.

\item Suppose that $y^2_{q}<\lceil
d_{q}^{\min}(\y^2)\rceil$ and  $y^2_{q-1}> x^l_{q-1}+1$. From {\bf
Step 4}, we find that
\[h_i(\y^2)=\left\{\begin{array}{ll}
y^2_i & \mbox{if $1\leq i\leq q-2$,}\\
y^2_{q-1}-1 & \mbox{if $i=q-1$,}\\
 x_i^u & \mbox{if $q\leq i\leq n$,}
 \end{array}\right.\]
$i=1,2,\ldots,n$. Since $y^1_i=y^2_i$, $i=1,2,\ldots,q-1$,
$y^1_{q}<y^2_q$, $\lceil d_{q}^{\min}(y^1)\rceil=\lceil
d_{q}^{\min}(y^2)\rceil$, and $k_2=q$, hence, we derive that
$k_1=k_2=q$ and $y^1_{q}<\lceil d_{q}^{\min}(y^1)\rceil$. Thus, it
follows from $y^1_{q-1}=y^2_{q-1}$ and {\bf Step 4} that
\[h_i(\y^1)=\left\{\begin{array}{ll}
y^1_i & \mbox{if $1\leq i\leq q-2$,}\\
y^1_{q-1}-1 & \mbox{if $i=q-1$,}\\
 x_i^u & \mbox{if $q\leq i\leq n$,}
 \end{array}\right.\]
$i=1,2,\ldots,n$. Therefore, $h(\y^1)=h(\y^2)$ and
$h(\y^1)\leq_lh(\y^2)$.
\end{enumerate}

\item[Case 3:] $2\leq k_2=q+1$. From the definition of $k_2$, we know that
$x_i^l<\lceil d_{i}^{\min}(\y^2)\rceil\leq y^2_i\leq\lfloor
d^{\max}_{i}(\y^2)\rfloor$, $i=1,2,\ldots,q$. Since $y^1_i=y^2_i$,
$i=1,2,\ldots,q-1$, hence, $\lceil d_{i}^{\min}(\y^1)\rceil=\lceil
d_{i}^{\min}(\y^2)\rceil$ and $\lfloor
d^{\max}_{i}(\y^1)\rfloor=\lfloor d^{\max}_{i}(\y^2)\rfloor$,
$i=1,2,\ldots,q$, $x_i^l<\lceil d_{i}^{\min}(\y^1)\rceil\leq
y^1_i\leq\lfloor d^{\max}_{i}(\y^1)\rfloor$, $i=1,2,\ldots,q-1$, and
$\lceil d_{q}^{\min}(\y^1)\rceil\leq\lfloor
d^{\max}_{q}(\y^1)\rfloor$.

\begin{enumerate}
\item Suppose that $y^2_{q+1}\geq\lceil
d_{q+1}^{\min}(\y^2)\rceil$, $\lfloor
d^{\max}_{q+1}(\y^2)\rfloor<\lceil d_{q+1}^{\min}(\y^2)\rceil$ and
$y^2_{q}\leq x^l_{q}+1$. From {\bf Step 4}, we find that
\[h_i(\y^2)=\left\{\begin{array}{ll}
y^2_i & \mbox{if $1\leq i\leq q-1$,}\\
 x^l_i & \mbox{if $q\leq i\leq n$,}
\end{array}\right.\]
$i=1,2,\ldots,n$. Since $y^1_q<y^2_q\leq x^l_{q}+1$,  we get that
$y^1_q=x^l_q$ and $k_1=q\geq 2$. Thus, it follows from
$y_q^1=x^l_q<\lceil d_{q}^{\min}(\y^1)\rceil$ and {\bf Step 4} that,
if $y^1_{q-1}\leq x_{q-1}^l+1$,
\[h_i(\y^1)=\left\{\begin{array}{ll}
y^1_i & \mbox{if $1\leq i\leq q-2$,}\\
 x^l_i & \mbox{if $q-1\leq i\leq n$,}
\end{array}\right.\]
$i=1,2,\ldots,n$, and if $y^1_{q-1}>x_{q-1}^l+1$,
\[h_i(\y^1)=\left\{\begin{array}{ll}
y^1_i & \mbox{if $1\leq i\leq q-2$,}\\
y^1_{q-1}-1 & \mbox{if $i=q-1$,}\\
 x_i^u & \mbox{if $q\leq i\leq n$,}
 \end{array}\right.\]
$i=1,2,\ldots,n$. Therefore, if $y^1_{q-1}\leq x_{q-1}^l+1$, then
 $h(\y^1)\leq_l h(\y^2)$ follows from $h_i(\y^1)=h_i(\y^2)$, $i=1,2,\ldots,q-2$, and $h_{q-1}(\y^1)=x_{q-1}^l<y^2_{q-1}=h_{q-1}(\y^2)$,
 and if $y^1_{q-1}>x_{q-1}^l+1$, then $h(\y^1)\leq_l h(\y^2)$ follows from $h_i(\y^1)=h_i(\y^2)$, $i=1,2,\ldots,q-2$, and
  $h_{q-1}(\y^1)=y_{q-1}^1-1<y^2_{q-1}=h_{q-1}(\y^2)$.

 \item Suppose that
$y^2_{q+1}\geq\lceil d_{q+1}^{\min}(\y^2)\rceil$, $\lfloor
d^{\max}_{q+1}(\y^2)\rfloor<\lceil d_{q+1}^{\min}(\y^2)\rceil$ and
$y^2_{q}> x^l_{q}+1$. From {\bf Step 4}, we find that
\[h_i(\y^2)=\left\{\begin{array}{ll}
y^2_i & \mbox{if $1\leq i\leq q-1$,}\\
y^2_{q}-1 & \mbox{if $i=q$,}\\
 x_i^u & \mbox{if $q+1\leq i\leq n$,}
 \end{array}\right.\]
$i=1,2,\ldots,n$.
\begin{itemize}
\item Assume that $\y^1\in P$. Thus, $h(\y^1)=\y^1$. Therefore,  $h(\y^1)\leq_l h(\y^2)$ follows
from $h_i(\y^1)=h_i(\y^2)$, $i=1,2,\ldots,q-1$, $h_q(\y^1)=y^1_{q}\leq
y^2_{q}-1=h_q(\y^2)$, and $h_i(\y^1)\leq x_i^u=h_i(\y^2)$,
$i=q+1,q+2,\ldots,n$.

\item Assume that $\y^1\notin P$. Then, we must have $k_1\geq
q$.

Consider that $k_1=q$. Since $\lceil
d_{q}^{\min}(\y^1)\rceil\leq\lfloor d^{\max}_{q}(\y^1)\rfloor$, hence,
 either $y^1_q>\lfloor d^{\max}_{q}(\y^1)\rfloor$ or $y^1_q<\lceil
d_{q}^{\min}(\y^1)\rceil$.
\begin{enumerate}
\item Suppose that $y^1_q>\lfloor d^{\max}_{q}(\y^1)\rfloor$. From {\bf Step 5}, we obtain that
\[h_i(\y^1)=\left\{\begin{array}{ll}
y^1_i & \mbox{if $1\leq i\leq q-1$,}\\
 \lfloor d^{\max}_{q}(\y^1)\rfloor & \mbox{if $i=q$,}\\
x^u_i & \mbox{if $q+1\leq i\leq n$,}
\end{array}\right.\]
$i=1,2,\ldots,n$. Thus, $h_q(\y^1)<y^1_q$. Therefore,
$h(\y^1)\leq_lh(\y^2)$ follows from $h_i(\y^1)=h_i(\y^2)$,
$i=1,2,\ldots,q-1$, and $h_q(\y^1)<y^1_q\leq y^2_q-1=h_q(\y^2)$.

\item Suppose that $y^1_q<\lceil
d_{q}^{\min}(\y^1)\rceil$. From {\bf Step 4}, we obtain that, if
$y^1_{q-1}\leq x_{q-1}^l+1$,
\[h_{i}(\y^1)=\left\{\begin{array}{ll}
y^1_i & \mbox{if $1\leq i\leq q-2$,}\\
x^l_i & \mbox{otheriwse,}
\end{array}\right.\]
$i=1,2,\ldots,n$, and if $y^1_{q-1}> x_{q-1}^l+1$,
\[h_i(\y^1)=\left\{\begin{array}{ll}
y^1_i & \mbox{if $1\leq i\leq q-2$,}\\
y^1_{q-1}-1 & \mbox{if $i=q-1$,}\\
 x_i^u & \mbox{if $q\leq i\leq n$,}
 \end{array}\right.\]
$i=1,2,\ldots,n$.  Therefore, if $y^1_{q-1}\leq x_{q-1}^l+1$, then
 $h(\y^1)\leq_l h(\y^2)$ follows from $h_i(\y^1)=h_i(\y^2)$,
$i=1,2,\ldots,q-2$, and
 $h_{q-1}(\y^1)=x^l_{q-1}<y^2_{q-1}=h_{q-1}(\y^2)$, and if $y^1_{q-1}>
 x_{q-1}^l+1$, then $h(\y^1)\leq_l h(\y^2)$ follows from $h_i(\y^1)=h_i(\y^2)$,
$i=1,2,\ldots,q-2$, and
 $h_{q-1}(\y^1)=y^1_{q-1}-1<y^2_{q-1}=h_{q-1}(\y^2)$.
\end{enumerate}

Consider that $k_1>q$. From the definition of $h(\y^1)$, we derive
that $h_i(\y^1)=y^1_i$, $i=1,2,\ldots,q-1$, and $h_q(\y^1)\leq y^1_q$.
Thus, $h(\y^1)\leq_l h(\y^2)$ follows immediately from
$h_i(\y^1)=h_i(\y^2)$, $i=1,2,\ldots,q-1$, $h_q(\y^1)\leq y^1_q\leq
y^2_q-1=h_q(\y^2)$, and $h_i(\y^1)\leq x_i^u=h_i(\y^2)$,
$i=q+1,q+2,\ldots,n$.

\end{itemize}

\item Suppose that $y^2_{q+1}\geq\lceil
d_{q+1}^{\min}(\y^2)\rceil$, $\lfloor
d^{\max}_{q+1}(\y^2)\rfloor\geq\lceil d_{q+1}^{\min}(\y^2)\rceil$ and
$y^2_{q+1}>\lfloor d^{\max}_{q+1}(\y^2)\rfloor$. From {\bf Step 5},
we find that
\[h_i(\y^2)=\left\{\begin{array}{ll}
y^2_i & \mbox{if $1\leq i\leq q$,}\\
 \lfloor d^{\max}_{q+1}(\y^2)\rfloor & \mbox{if $i=q+1$,}\\
x^u_i & \mbox{if $q+2\leq i\leq n$,}
\end{array}\right.\]
$i=1,2,\ldots,n$.
\begin{itemize}
\item
Assume that $\y^1\in P$. Then, $h(\y^1)=\y^1$.
 Thus, from $y^1_q<y^2_q$, we obtain that $h_{q}(\y^1)<y^2_q=h_{q}(\y^2)$. Therefore,
 $h(\y^1)\leq_lh(\y^2)$ follows immediately from $h_i(\y^1)=h_i(\y^2)$, $i=1,2,\ldots,q-1$, and $h_q(\y^1)<h_q(\y^2)$.

\item Assume that $\y^1\notin P$. Then, we must have $k_1\geq q$.

Consider that $k_1=q$. Since $\lceil
d_{q}^{\min}(\y^1)\rceil\leq\lfloor d^{\max}_{q}(\y^1)\rfloor$, hence,
 either $y^1_q>\lfloor d^{\max}_{q}(\y^1)\rfloor$ or $y^1_q<\lceil
d_{q}^{\min}(\y^1)\rceil$.
\begin{enumerate}
\item Suppose that $y^1_q>\lfloor d^{\max}_{q}(\y^1)\rfloor$. From {\bf Step 5}, we obtain that
\[h_i(\y^1)=\left\{\begin{array}{ll}
y^1_i & \mbox{if $1\leq i\leq q-1$,}\\
 \lfloor d^{\max}_{q}(\y^1)\rfloor & \mbox{if $i=q$,}\\
x^u_i & \mbox{if $q+1\leq i\leq n$,}
\end{array}\right.\]
$i=1,2,\ldots,n$. Thus, $h_q(\y^1)<y^1_q$. Therefore,
$h(\y^1)\leq_lh(\y^2)$ follows from $h_i(\y^1)=h_i(\y^2)$,
$i=1,2,\ldots,q-1$, and $h_q(\y^1)<y^1_q<y^2_q=h_q(\y^2)$.

\item Suppose that $y^1_q<\lceil
d_{q}^{\min}(\y^1)\rceil$. From {\bf Step 4}, we obtain that, if
$y^1_{q-1}\leq x_{q-1}^l+1$,
\[h_{i}(\y^1)=\left\{\begin{array}{ll}
y^1_i & \mbox{if $1\leq i\leq q-2$,}\\
x^l_i & \mbox{otheriwse,}
\end{array}\right.\]
$i=1,2,\ldots,n$, and if $y^1_{q-1}> x_{q-1}^l+1$,
\[h_i(\y^1)=\left\{\begin{array}{ll}
y^1_i & \mbox{if $1\leq i\leq q-2$,}\\
y^1_{q-1}-1 & \mbox{if $i=q-1$,}\\
 x_i^u & \mbox{if $q\leq i\leq n$,}
 \end{array}\right.\]
$i=1,2,\ldots,n$.  Therefore, if $y^1_{q-1}\leq x_{q-1}^l+1$, then
 $h(\y^1)\leq_l h(\y^2)$ follows from $h_i(\y^1)=h_i(\y^2)$,
 $i=1,2,\ldots,q-2$, and
 $h_{q-1}(\y^1)=x_{q-1}^l<y^1_{q-1}=y^2_{q-1}=h_{q-1}(\y^2)$, and if $y^1_{q-1}> x_{q-1}^l+1$, then $h(\y^1)\leq_l h(\y^2)$ follows
 from $h_i(\y^1)=h_i(\y^2)$,
 $i=1,2,\ldots,q-2$, and
 $h_{q-1}(\y^1)=y_{q-1}^1-1<y^1_{q-1}=y^2_{q-1}=h_{q-1}(\y^2)$.
\end{enumerate}
Consider that $k_1>q$. From the definition of $h(\y^1)$, we derive
that $h_i(\y^1)=y^1_i$, $i=1,2,\ldots,q-1$, and $h_q(\y^1)\leq y^1_q$.
Thus, $h(\y^1)\leq_l h(\y^2)$ follows immediately from
$h_i(\y^1)=h_i(\y^2)$, $i=1,2,\ldots,q-1$, and $h_q(\y^1)\leq
y^1_q<y^2_q=h_q(\y^2)$.
\end{itemize}

\item Suppose that $y^2_{q+1}<\lceil
d_{q+1}^{\min}(\y^2)\rceil$ and  $y^2_{q}\leq x^l_{q}+1$. From {\bf
Step 4}, we find that
\[h_i(y^2)=\left\{\begin{array}{ll}
y^2_i & \mbox{if $1\leq i\leq q-1$,}\\
 x^l_i & \mbox{if $q\leq i\leq n$,}
\end{array}\right.\]
$i=1,2,\ldots,n$. Since $y^1_q<y^2_q\leq x^l_{q}+1$,  we get that
$y^1_q=x^l_q$ and $k_1=q\geq 2$. Thus, we obtain from {\bf Step 4}
that, if $y^1_{q-1}\leq x_{q-1}^l+1$,
\[h_i(\y^1)=\left\{\begin{array}{ll}
y^1_i & \mbox{if $1\leq i\leq q-2$,}\\
 x^l_i & \mbox{if $q-1\leq i\leq n$,}
\end{array}\right.\]
$i=1,2,\ldots,n$, and if $y^1_{q-1}>x_{q-1}^l+1$,
\[h_i(\y^1)=\left\{\begin{array}{ll}
y^1_i & \mbox{if $1\leq i\leq q-2$,}\\
y^1_{q-1}-1 & \mbox{if $i=q-1$,}\\
 x_i^u & \mbox{if $q\leq i\leq n$,}
 \end{array}\right.\]
$i=1,2,\ldots,n$. Therefore, if $y^1_{q-1}\leq x_{q-1}^l+1$, then
 $h(\y^1)\leq_l h(\y^2)$ follows from $h_i(\y^1)=h_i(\y^2)$, $i=1,2,\ldots,q-2$, and
 $h_{q-1}(\y^1)=x_{q-1}^l<y^1_{q-1}=y^2_{q-1}=h_{q-1}(\y^2)$, and if
 $y^1_{q-1}>x_{q-1}^l+1$, then $h(\y^1)\leq_l h(\y^2)$ follows from $h_i(\y^1)=h_i(\y^2)$, $i=1,2,\ldots,q-2$, and
 $h_{q-1}(\y^1)=y_{q-1}^1-1<y^1_{q-1}=y^2_{q-1}=h_{q-1}(\y^2)$.

\item Suppose that $y^2_{q+1}<\lceil
d_{q+1}^{\min}(\y^2)\rceil$ and  $y^2_{q}> x^l_{k_2-1}+1$. From {\bf
Step 4}, we find that
\[h_i(\y^2)=\left\{\begin{array}{ll}
y^2_i & \mbox{if $1\leq i\leq q-1$,}\\
y^2_{q}-1 & \mbox{if $i=q$,}\\
 x_i^u & \mbox{if $q+1\leq i\leq n$,}
 \end{array}\right.\]
$i=1,2,\ldots,n$.

\begin{itemize}
\item Assume that $\y^1\in P$. Thus, $h(\y^1)=\y^1$. Therefore, $h(\y^1)\leq_l h(\y^2)$ follows
from $h_i(\y^1)=h_i(\y^2)$, $i=1,2,\ldots,q-1$, $h_q(\y^1)=y^1_{q}\leq
y^2_{q}-1=h_q(\y^2)$, and $h_i(\y^1)\leq x_i^u=h_i(\y^2)$,
$i=q+1,q+2,\ldots,n$.

\item Assume that $\y^1\notin P$. Then, we must have that $k_1\geq
q$.

Consider that $k_1=q$. Since $\lceil
d_{q}^{\min}(\y^1)\rceil\leq\lfloor d^{\max}_{q}(\y^1)\rfloor$, hence,
 either $y^1_q>\lfloor d^{\max}_{q}(\y^1)\rfloor$ or $y^1_q<\lceil
d_{q}^{\min}(\y^1)\rceil$.
\begin{enumerate}
\item Suppose that $y^1_q>\lfloor d^{\max}_{q}(\y^1)\rfloor$. From {\bf Step 5}, we obtain that
\[h_i(\y^1)=\left\{\begin{array}{ll}
y^1_i & \mbox{if $1\leq i\leq q-1$,}\\
 \lfloor d^{\max}_{q}(\y^1)\rfloor & \mbox{if $i=q$,}\\
x^u_i & \mbox{if $q+1\leq i\leq n$,}
\end{array}\right.\]
$i=1,2,\ldots,n$. Thus, $h_q(\y^1)<y^1_q$. Therefore,
$h(\y^1)\leq_lh(\y^2)$ follows from $h_i(\y^1)=h_i(\y^2)$,
$i=1,2,\ldots,q-1$, and $h_q(\y^1)<y^1_q\leq y^2_q-1=h_q(\y^2)$.

\item Suppose that $y^1_q<\lceil
d_{q}^{\min}(\y^1)\rceil$. From {\bf Step 4}, we obtain that, if
$y^1_{q-1}\leq x_{q-1}^l+1$,
\[h_{i}(\y^1)=\left\{\begin{array}{ll}
y^1_i & \mbox{if $1\leq i\leq q-2$,}\\
x^l_i & \mbox{otheriwse,}
\end{array}\right.\]
$i=1,2,\ldots,n$, and if $y^1_{q-1}> x_{q-1}^l+1$,
\[h_i(\y^1)=\left\{\begin{array}{ll}
y^1_i & \mbox{if $1\leq i\leq q-2$,}\\
y^1_{q-1}-1 & \mbox{if $i=q-1$,}\\
 x_i^u & \mbox{if $q\leq i\leq n$,}
 \end{array}\right.\]
$i=1,2,\ldots,n$.  Therefore, if $y^1_{q-1}\leq x_{q-1}^l+1$, then
 $h(\y^1)\leq_l h(\y^2)$ follows from $h_i(\y^1)=h_i(\y^2)$,
$i=1,2,\ldots,q-2$, and
 $h_{q-1}(\y^1)=x_{q-1}^l<y^1_{q-1}=y^2_{q-1}=h_{q-1}(\y^2)$, and if $y^1_{q-1}>
 x_{q-1}^l+1$, then
 $h(\y^1)\leq_l h(\y^2)$ follows from $h_i(\y^1)=h_i(\y^2)$,
$i=1,2,\ldots,q-2$, and
 $h_{q-1}(\y^1)=y_{q-1}^1-1<y^1_{q-1}=y^2_{q-1}=h_{q-1}(\y^2)$.
\end{enumerate}

Consider that $k_1>q$. From the definition of $h(\y^1)$, we derive
that $h_i(\y^1)=y^1_i$, $i=1,2,\ldots,q-1$, and $h_q(\y^1)\leq y^1_q$.
Thus, $h(\y^1)\leq_l h(\y^2)$ follows immediately from
$h_i(\y^1)=h_i(\y^2)$, $i=1,2,\ldots,q-1$, $h_q(\y^1)\leq y^1_q\leq
y^2_q-1=h_q(\y^2)$, and $h_i(\y^1)\leq x_i^u=h_i(\y^2)$,
$i=q+1,q+2,\ldots,n$.

\end{itemize}
\end{enumerate}

\item[Case 4:] $k_2>q+1$. From $k_2-1>q$, we obtain that
$h_i(\y^2)=y^2_i$, $i=1,2,\ldots,q$. Thus, $\y^1\leq_lh(\y^2)$ since
$y^1_i=y^2_i$, $i=1,2,\ldots,q-1$, and $y^1_q<y^2_q$. Therefore, it
follows immediately from Lemma~\ref{lessy} that
$h(\y^1)\leq_lh(\y^2)$.
\end{description} Since $2\leq k_2\leq
n$, hence, one of the above four cases must occur. The above results
show that, for every case, it always holds that
$h(\y^1)\leq_lh(\y^2)$. This completes the proof.\hfill\fbox{}

From Definition~\ref{imd}, one can see that, for each $\y\in D(P)$,
it takes at most $2n$ linear programs to compute $h(\y)$. Therefore,
$h(\y)$ is polynomial-time defined for any given $\y\in D(P)$.

As a corollary of Theorem~\ref{im} and Theorem~\ref{npcomplete}, we obtain that
\begin{corollary}
Given lattice $(L,\leq_l)$ and an order preserving mapping $f$ as a polynomial function,
determining that $f$ has a unique fixed point in $L$ is a Co-NP hard
problem.
\end{corollary}

\subsection{Poof of Co-NP-hard in Componentwise Ordering}

Let $N=\{1,2,\ldots,n\}$ and $N_0=\{0,1,\ldots,n\}$.  For any real number
$\alpha$, let $\lfloor\alpha\rfloor$ denote the greatest integer
less than or equal to $\alpha$ and $\lceil\alpha\rceil$ the smallest
integer greater than or equal to $\alpha$. For any vector
$\x=(x_1,x_2,\ldots,x_n)^{\top}\in R^n$, let $\lfloor
\x\rfloor=(\lfloor x_1\rfloor, \lfloor x_2\rfloor,\ldots,\lfloor
x_n\rfloor)^{\top}$ and $\lceil \x\rceil=(\lceil x_1\rceil,\lceil
x_2\rceil,\ldots,\lceil x_n\rceil)^{\top}$. Given these notations,
we present a polynomial-time reduction of integer programming, which
is as follows.

For any $\x\in R^n$, let \[P(\x)=\{\y\in P\;|\;\y\leq_c \x\}.\]
Then, as a direct result of the property of the matrix $A$, one can easily
obtain that
\begin{lemma} \label{maxclosed}For any given $\x\in R^n$, if
$\x^{1}=(x_{1}^{1},x^{1}_{2},\ldots,x_{n}^{1})^{\top}\in P(\x)$ and
$\x^{2}=(x_{1}^{2},x^{2 }_{2},\ldots,x_{n}^{2})^{\top}\in P(\x)$, then
\[\bar{\x}=\max(\x^1,\x^2)=(\max\{x^{1}_{1},x^{2}_{1}\},\max\{x^{1}_{2},x^{2}_{2}\},
\ldots,\max\{x^{1}_{n},x^{2}_{n}\})^{\top} \in P(\x).\]
\end{lemma}

Let $\e=(1,1,\ldots,1)^{\top}\in R^n$. For any given $\bv\in R^n$, if
$P(\bv)\neq\emptyset$, Lemma~\ref{maxclosed} implies that $\max_{\x\in
P(\bv)}\e^{\top}\x$ has a unique solution, which we denote by
$\x^\bv=(x_1^\bv,x_2^\bv,\ldots,x_n^\bv)^{\top}$.
\begin{lemma}\label{max} $\x\leq_c \x^\bv$ for all $\x\in P(\bv)$.
\end{lemma}
{\bf Proof.} Suppose that there is a point
$\x^0=(x_1^0,x_2^0,\ldots,x_n^0)^{\top}\in P(\bv)$ with $x^0_k>x^\bv_k$
for some $k\in N$. Then, Lemma~\ref{maxclosed} implies that
\[\x^{\bv0}=(\max\{x^0_1,x^\bv_1\},\max\{x^0_2,x^\bv_2\},\ldots,\max\{x^0_n,x^\bv_n\})^{\top}\in
P(\bv).\] Thus, $e^{\top}\x^{0\bv}>e^{\top}\x^\bv=\max_{\x\in
P(\bv)}\e^{\top}\x$. A contradiction arises. This completes the
proof.\hfill\fbox{}

Let $\x^{\max}=(x^{\max}_1,x^{\max}_2,\ldots,x_n^{\max})^{\top}$ be
the unique solution of $\max_{\x\in P}\e^{\top}\x$ and $\x^{\min}=(x^{\min}_1,x^{\min}_2,$\\$\ldots,x^{\min}_n)^{\top}$ with
$x^{\min}_j=\min_{\x\in P}x_j$, $j=1,2,\ldots,n$. Then, $\x^{\min}\leq_c
\x\leq_c \x^{\max}$ for all $\x\in P$.  Let
\[D(P)=\{\x\in {Z}^n\;|\;\x^l\leq_c \x\leq_c \x^u\},\]
where \(\x^u=\lfloor \x^{\max}\rfloor\) and \(\x^l=\lfloor
\x^{\min}\rfloor.\) Thus, $D(P)$ contains all integer points in $P$.
Without loss of generality, we assume that $\x^l<_c\x^{\min}$ (Let
$x_i^l=x_i^{\min}-1$ if $x_i^l=x_i^{\min}$ for some $i\in N$).
Obviously, the sizes of both $\x^l$ and $\x^u$ are bounded by
polynomials of the sizes of the matrix $A$ and the vector $\bb$ since $\x^l$ and $\x^u$ are
obtained from the solutions of linear programs with rational data.

For $\x\in R^n$, we define \[h(\x)=\lfloor d(\x)\rfloor\] with
\[d(\x)=\left\{\begin{array}{ll}
\x^{l} & \mbox{if $P(\x)=\emptyset$,}\\
\\
\mbox{argmax}_{\y\in P(\x)}\e^{\top}\y & \mbox{otherwise.}
\end{array}\right.\]
It follows  from Lemma~\ref{max} that $d(\x)$ is well defined.
\begin{example}
\label{example1} Consider $P=\{\x\in R^3\;|\; A\x\leq_c \bb\}$, where
\[A=\left(\begin{array}{ccc}
     2  &  -1 &    0\\
    -1   &  3  &   0\\
     0  &   0  &   2\\
     0  &  -1  &  -1
\end{array}\right)\]
and $\bb =(
     0,
   -10,
    10,
     0
     )^{\top}$. For $\y=(-3,
    -4,
    5)^{\top}$, $h(\y)=(-3, -5, 5)^{\top}$. An illustration
of $h$ can be found in Fig.\ref{cfig1}.

\begin{figure}[H]
\begin{center}
\hspace{0mm} \epsfig{figure=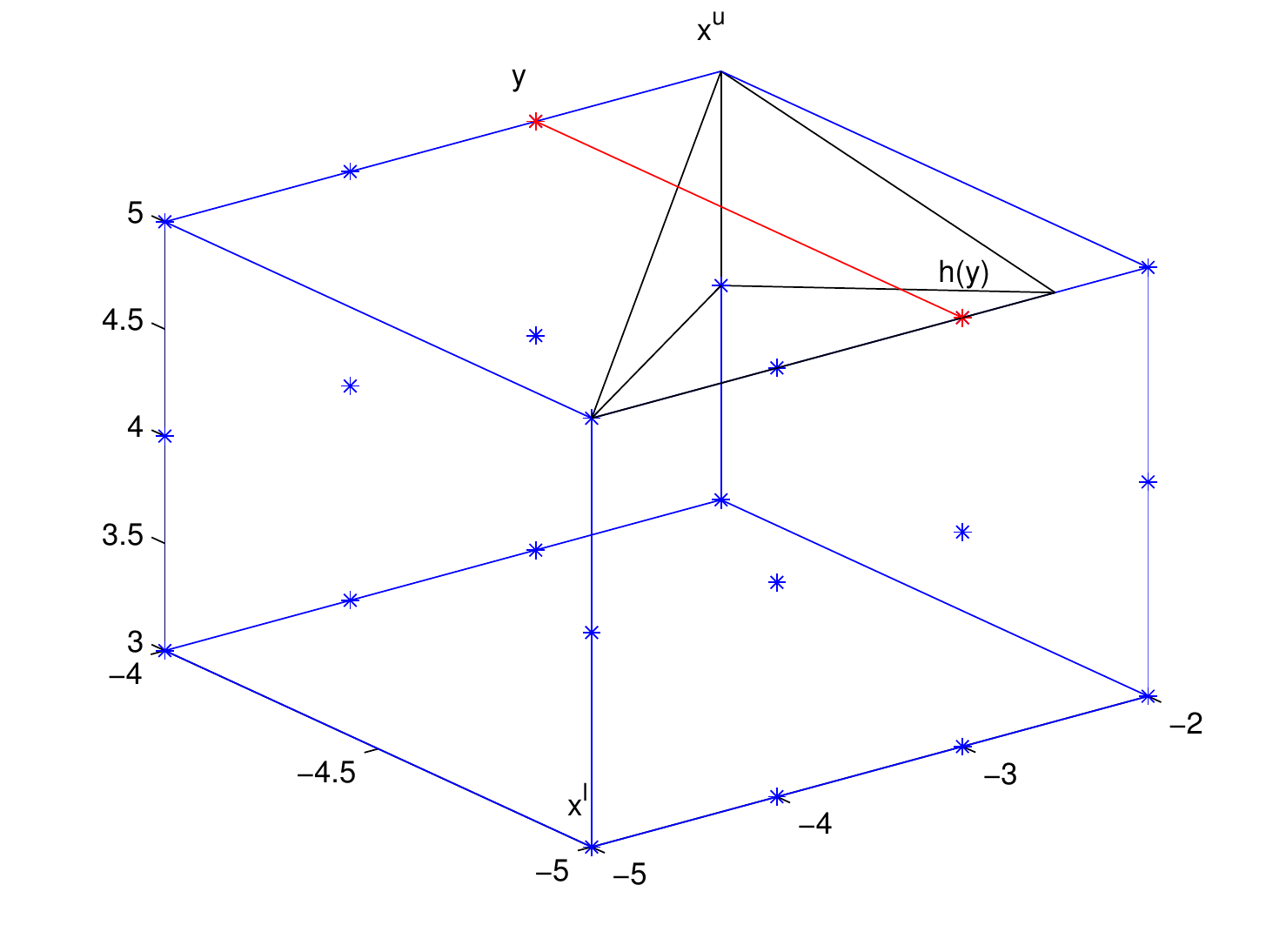,width=12cm,height=12cm}
\end{center}
\caption{\label{cfig1}An Illustration of $h$}
\end{figure}
\end{example}

\begin{lemma}
\label{ip} $h$ is an order preserving mapping from $R^n$ to $D(P)$.
Moreover, $h(\x^*)=\x^*\neq \x^{l}$ if and only if $\x^*$ is an integer
point in $P$.
\end{lemma}
{\bf Proof.} Let $\x^1$ and $\x^2$ be two different points of $R^n$
with $\x^1\leq_c \x^2$. Then, $P(\x^1)\subseteq P(\x^2)$. Thus, from the
definition of $d(\x)$, we obtain that $\x^{\min}\leq_c d(\x^1)\leq_c
d(\x^2)\leq_c \x^{\max}$. The first part of the lemma follows
immediately.

Let $\x^*$ be an integer point in $P$. Then,
\[d(\x^*)=\mbox{argmax}_{\y\in P(\x^*)}\e^{\top}\y=\x^*.\] Thus, $h(\x^*)=\x^*\neq \x^l$.

 Let $\x^*$ be a
point in $R^n$ satisfying that $h(\x^*)=\x^*\neq \x^{l}$. Suppose that
$P(\x^*)=\emptyset$. Then, $d(\x^*)=\x^l$. Thus,
\[\x^*=h(\x^*)=\lfloor \x^l\rfloor=\x^l.\] A contradiction
occurs. Therefore, $P(\x^*)\neq\emptyset$ and, consequently,
$d(\x^*)\in P$. Since
\[\x^*\geq_c d(\x^*)\geq_c\lfloor d(\x^*)\rfloor = h(\x^*)=\x^*,\]hence,
$d(\x^*)=\lfloor d(\x^*)\rfloor=\x^*$. This completes the
proof.\hfill\fbox{}

Let $(L,\leq_c)$ be a finite lattice and $f$ an order preserving mapping
from $L$ into itself.
  As a corollary of
Theorem~\ref{npcomplete} and Lemma~\ref{ip}, we obtain that
\begin{corollary} \label{nphard} Given lattice $(L,\leq_c)$ and an order preserving mapping $f$ as a polynomial function,
determining that $f$ has a unique fixed point in $L$ is a Co-NP hard
problem.
\end{corollary}

\end{document}